\documentclass[journal,twoside,web]{ieeecolor} % Ensure ieeecolor.cls is in the directory or use IEEEtran as fallback

\usepackage{generic}
\usepackage{cite}
\usepackage{amsmath,amssymb,amsfonts}
 \usepackage{algorithm}
\usepackage{algorithmicx}
\usepackage{algpseudocode}
\usepackage{graphicx}
\usepackage{subcaption}
\usepackage{textcomp}
\usepackage{bbm}
\usepackage{dsfont}

\newtheorem{theorem}{\bf Theorem}

\newtheorem{problem}{\bf Problem}
\newtheorem{remark}{\bf Remark}
\newtheorem{definition}{\bf Definition}

\newtheorem{assumption}{\bf Assumption}

\def\QED{~\rule[-1pt]{5pt}{5pt}\par\medskip}
\begin{document}
%\title{Minimum-Information Kalman-Bucy Filtering in Continuous Time}
\title{Path Integral Methods for Synthesizing and Preventing Stealthy Attacks in Nonlinear Cyber-Physical Systems}
\author{Apurva Patil, Kyle Morgenstein, Luis Sentis, and Takashi Tanaka
\thanks{A. Patil is with the Walker Department of Mechanical Engineering, University of Texas at Austin, Austin, TX 78712 USA (e-mail: apurvapatil@utexas.edu). K. Morgenstein and L. Sentis are with the Department of Aerospace Engineering and Engineering Mechanics at the University of Texas at Austin, Austin, TX 78712 USA (e-mail: kylem@utexas.edu, lsentis@utexas.edu). T. Tanaka is with Aeronautics and Astronautics, and of Electrical and Computer Engineering departments at Purdue University, West Lafayette, IN 47907 USA (e-mail: tanaka16@purdue.edu).}
% \thanks{T. Cuvelier is with the Chandra Department of Electrical and Computer Engineering at the University of Texas at Austin,  Austin, TX 78712 USA (e-mail: tcuvelier@utexas.edu). T. Tanaka is with the Department of Aerospace Engineering and Engineering Mechanics at the University of Texas at Austin, Austin, TX 78712 USA (e-mail: ttanaka@utexas.edu). }
}

\maketitle

\begin{abstract}
This paper studies the synthesis and mitigation of stealthy attacks in nonlinear cyber-physical systems (CPS). To quantify stealthiness, we employ the Kullback–Leibler (KL) divergence, a measure rooted in hypothesis testing and detection theory, which captures the trade-off between an attacker’s desire to remain stealthy and her goal of degrading system performance. First, we synthesize the worst-case stealthy attack in nonlinear CPS using the path integral approach. Second, we consider how a controller can mitigate the impact of such stealthy attacks by formulating a minimax KL control problem, yielding a zero-sum game between the attacker and the controller. Again, we leverage a path integral–based solution that computes saddle-point policies for both players through Monte Carlo simulations. We validate our approach using unicycle navigation and cruise control problems, demonstrating how an attacker can covertly drive the system into unsafe regions, and how the controller can adapt her policy to combat the worst-case attacks. 
\end{abstract}

\begin{IEEEkeywords}
Stealthy attack, Continuous time systems, Nonlinear cyber-physical systems, Kullback-Leibler (KL) control, Path integral control. 
%visit \underline{http://www.ieee.org/organizations/pubs/ani\_prod/keywrd98.txt}
\end{IEEEkeywords}

\section{Introduction}
\label{sec:introduction}
% The work presented in this paper is inspired by the security of cyber-physical systems (CPS) against stealthy attacks. CPS are the next generation of engineered systems that tightly integrate computation, communication, control, and physical processes \cite{kim2012cyber}. Due to the interconnection of different technologies and components, CPS are vulnerable to adversarial intrusion which may cause severe consequences on the national economy, social security, or even loss of human lives \cite{poovendran2011special}. Recently reported cyber attacks, e.g., the StuxNet malware \cite{karnouskos2011stuxnet}, and the Maroochy water bleach \cite{slay2007lessons}, evidently indicate that security is of fundamental importance to ensure the safe operation of CPS. \par 

The work presented in this paper is motivated by the need to protect cyber-physical systems (CPS) from stealthy attacks. As the next generation of engineered systems, CPS integrate computation, communication, control, and physical processes in a tightly interconnected manner \cite{kim2012cyber}. Because of this close integration across diverse technologies, CPS are susceptible to adversarial intrusions that can lead to severe ramifications for national economies, public safety, and even human lives \cite{poovendran2011special}. Notable cyber attacks—such as the StuxNet malware \cite{karnouskos2011stuxnet} and the Maroochy water bleach incident \cite{slay2007lessons}—underscore the critical importance of ensuring robust security measures for the safe operation of CPS.\par

With the increasing adoption of CPS, attack strategy and defense mechanism design have received considerable attention in the literature. In this paper, we consider worst-case attack synthesis and its mitigation problems for nonlinear continuous-time CPS and propose their solutions using the path integral control method. 
% Suppose a CPS is operated by a controller in a stochastic environment. Assume that a stealthy attacker hijacks the control authority and injects an attack signal to misguide the system. The attacker carefully designs the attack signals such that they can be disguised as the effects of natural disturbances in the CPS, such as controller noises and environmental disturbances.
Specifically, we consider a scenario in which a controller operates a CPS in a stochastic environment while being vulnerable to potential adversarial intrusion. A stealthy attacker hijacks the control input and injects malicious signals that are carefully designed to resemble natural disturbances—such as sensor noise or environmental fluctuations—thereby evading detection.
Meanwhile, the controller continuously monitors the system to distinguish the attack signals from natural background noises by designing an appropriate detection test. However, knowing that the system is continuously monitored, a rational attacker will conduct a covert attack, maximizing the attack’s impact while avoiding detection. Intuitively, there is a tradeoff between the performance degradation an attacker can induce and how easy it is to detect the attack \cite{teixeira2012attack}. We solve two key problems in this paper. First, we compute the worst-case stealthy attack signal assuming that the controller's policy is fixed and is known to the attacker. Next, we design the controller's policy to mitigate the risk of stealthy attacks, which will serve as the physical watermarking policy that countermeasures the potential stealthy attacks.\par

We design a stealthy attack synthesis problem motivated by hypothesis testing. Despite recent progress, existing applications of hypothesis testing frameworks to control systems \cite{bai2017data,kung2016performance,guo2018worst, shang2021worst} are restricted to discrete-time linear-Gaussian settings. The assumption of linear systems driven by Gaussian noises provides a significant advantage in designing and analyzing the statistical tests for detecting an attack, and the removal of such assumptions is marked as an open problem in \cite{sandberg2022secure}. Invoking Stein’s lemma, Bai et al. \cite{bai2017data} introduced the notion of $\epsilon$-stealthiness as measured by relative entropy. The worst-case degradation of linear control systems attainable by $\epsilon$-stealthy attacks was studied in \cite{guo2018worst, shang2021worst}. In this paper, we generalize this analysis to the class of nonlinear continuous-time dynamical systems. We then propose the path integral approach for the synthesis of stealthy attack policies. Our preliminary result \cite{patil2023simulator} demonstrates the feasibility of such a path-integral-based attack synthesis in a discrete-time setting.\par

Next, we consider the minimax formulation in which the controller is also allowed to inject a control input into the system to mitigate the attack impact. Several mitigation strategies against stealthy attacks have been proposed in the literature \cite{sandberg2022secure}. The trade-off between control performance and system security was investigated under a stochastic game framework in \cite{miao2013stochastic}. The minimax games between the attacker and the controller have been studied by many authors in the systems and control community.  By formulating a zero-sum game, Zhang and Venkitasubramaniam [8] studied false data injection and detection problems
in infinite-horizon linear-quadratic-Gaussian systems. The works \cite{bai2017data, guo2018worst, shang2021worst} adopted
the hypothesis testing theory to characterize covert false data injection attacks against control systems. Despite recent progress, existing applications of hypothesis testing frameworks to control systems are limited to linear discrete-time settings. The goal of this paper is to broaden the scope of the literature by formulating the aforementioned minimax game in continuous time for nonlinear cyber-physical systems. One effort in this direction is presented in one of our earlier works \cite{tanaka2024covert}. 
% Furthermore, we propose a simulator-driven approach to solve this minimax game in continuous time for nonlinear CPS.

% In the security framework,  \cite{kung2016performance,bai2017data} study the detectability of an attacker in a stochastic control setting. Similar to our formulation, \cite{kung2016performance,bai2017data} provide a KL divergence-based optimization problem. KL divergence objective is also used in reinforcement learning~\cite{schulman2015trust,filippi2010optimism} to improve the learning performance and in KL control frameworks~\cite{todorov2007linearly,ito2022kullback} for the efficient computation of optimal policies. In \cite{ito2022kullback}, It\^o et al. studied the KL control problem for nonlinear continuous-state space systems and characterized the optimal policies. Different from \cite{ito2022kullback}, we provide a randomized control algorithm based on path integral approach that converges to the optimal policy as the number of samples increases.

The conventional approaches for attack synthesis and mitigation necessitate explicit models of physical systems. However, in increasingly common scenarios, physical models are represented by various forms of ``digital twins," such as trained neural ODEs, which are easy to simulate but are not necessarily easy to express as models in the classical sense. This has elevated the importance of simulator-based CPS interaction, where the agent can directly use real-time simulation data to assist her decision-making without needing to construct an explicit model. In this paper, we use a specific type of simulator-based control scheme known as the \textit{path integral control method}. Path integral control is a sampling-based algorithm employed to solve nonlinear stochastic optimal control problems numerically ~\cite{kappen2005path, theodorou2010generalized, williams2016aggressive}. It allows the policy designer to compute the optimal control inputs online using Monte Carlo samples of system paths. Such an algorithm was pioneered by Kappen \cite{kappen2005path} and has been generalized in the robotics and machine-learning literature \cite{theodorou2010generalized,williams2016aggressive}. The Monte Carlo simulations can be massively parallelized on GPUs, and thus the path integral approach is less susceptible to the curse of dimensionality \cite{williams2017model}. \par

The contributions of this paper are as follows:
\subsection*{Contributions}
\begin{enumerate}
\item Kullback–Leibler (KL) control formulation for stealthy attack synthesis: We propose a KL control framework that models the trade-off between attack impact and detectability (Problem \ref{prob: KL}). This formulation is subsequently transformed into an equivalent quadratic-cost stochastic optimal control (SOC) problem  (Theorem \ref{thm:prob1_1}), enabling the synthesis of worst-case stealthy attacks.

\item Path-integral-based stealthy attack synthesis: We develop a novel path-integral-based method to synthesize worst-case stealthy attacks in real time for nonlinear continuous-time systems (Theorem \ref{Theorem: sol of SOC}). Notably, the proposed approach does not require an explicit system model or an explicit policy synthesis.

\item Zero-sum game formulation for attack mitigation: In order to mitigate the risk of stealthy attacks, we propose a novel zero-sum game formulation to model the competition between the attacker and the controller (Problem \ref{prob: minimax_KL}). We further show that this formulation is equivalent to both a risk-sensitive control problem and an H$_\infty$ control problem (Problems \ref{prob: risk-sensitive control} and \ref{prob: game}, respectively).  

\item Path-integral-based attack mitigation:  We provide a path-integral-based approach that allows the controller to synthesize the attack mitigating control inputs online, without relying on explicit models or policy synthesis (Theorems \ref{thm: risk-sensitive control} and \ref{thm: two-player game}).

% \item Connection between risk-sensitive control and two-player zero-sum stochastic differential game: As a byproduct of our analysis, we show that the path integral solution to a certain class of risk-sensitive control problems coincides with that of two-player zero-sum stochastic differential game (Remark \ref{rem: equivalence}).

\end{enumerate}

% \subsection*{Interpretations}
% \begin{itemize}
% \item Problem formulation in this paper is inspired by CPS security against stealthy attacks.
% \item Solution $u_t$ obtained by solving the risk-sensive control can be viewed as the optimal watermarking policy. 
% \item As a motivation/simulation example, how about considering a remote UAV control scenario, where a stealthy attacker may hijack the control authority? $u_t$ in this scenario is the control input, but at the same time, it also serves as the physical watermarking signal that countermeasures the potential stealthy attack. 
% \end{itemize}

% \subsection*{Organization of This Paper}
% Next, we apply the path integral method to demonstrate the possibility of a sophisticated simulator-driven attack synthesis. Finally, we consider the minimax formulation in which the agent is also allowed to inject a control input into the system to mitigate the attack impact.

\section{Preliminaries}
\label{sec:formulation}

\subsection{Notation}
Given two probability measures $P$ and $Q$ on a measurable space $(\Omega, \mathcal{F})$, the Kullback-Leibler (KL) divergence from $Q$ to $P$ is defined as
$D(P\|Q)=\int_\Omega \log\frac{dP}{dQ}(\omega)P(d\omega)$
if the Radon-Nikodym derivative $\frac{dP}{dQ}$ exists, and $D(P\|Q)=+\infty$ otherwise. $\mathcal{N}(\mu, \Sigma)$ represents the Gaussian distribution with mean $\mu$ and covariance $\Sigma$.
% For a random variable $x:(\Omega,\mathcal{F})\rightarrow \mathcal{X}$, the probability distribution of $x$ under $P$ is defined by $P_x(dx)=P\left\{\omega\in\Omega: x(\omega)\in dx\right\}$. 
% For notational simplicity, we will write $P(x)$ instead of $P_x(dx)$.
Throughout this paper, we use the natural logarithm.

\subsection{System Setup}
{\color{black}
All the random processes considered in this paper are defined on the probability space $(\Omega, \mathcal{F}, P)$. Let $w_t$ be an $m$-dimensional standard Brownian motion with respect to the probability measure $P$, and let $\mathcal{F}_t\subset \mathcal{F}$ be an increasing family of $\sigma$-algebras such that $w_t$ is a martingale with respect to $\mathcal{F}_t$. In the sequel, $w_t$ will be used to model the natural disturbance. Consider the class of continuous-time nonlinear cyber-physical systems (shown in Figure~\ref{fig:cps_security})  with state $x_t\in\mathbb{R}^n$, control input $u_t\in\mathbb{R}^\ell$, and the input random process $v_t\in\mathbb{R}^m$:
\begin{equation}\label{SDE}
dx_t=f_t(x_t) dt +g_t(x_t)u_tdt+h_t(x_t)dv_t.
\end{equation}
Under the nominal operation condition, we assume $v_t=w_t$ for all $0\leq t\leq T$, i.e.,  the natural disturbance directly enters into the system.
In this paper, we are concerned with the competition between the controller agent (showing blue in Figure~\ref{fig:cps_security}) and the attacker agent (showing red in Figure~\ref{fig:cps_security}) over the dynamical system \eqref{SDE}.

\subsubsection{Controller}
The controller has a legitimate authority to operate the system \eqref{SDE}. Her primary role is to compute the control input $u_t$ to minimize the expected cost of the system operation, which is denoted by $\mathbb{E}\left[\int_0^T c_t(x_t, u_t)dt\right]$ with some measurable functions $c(\cdot,\cdot)$. We assume that the controller applies a state feedback policy. Specifically, we impose a restriction $u \in \mathcal{U}$ where
\[
\mathcal{U}=\{u: \text{$u_t$ is $\mathcal{F}_t^x$-adapted It\^o process}\}
\]
and  $\mathcal{F}_t^x=\sigma(\{x_s:0\leq s\leq t\})\subset \mathcal{F}$ is a sigma algebra generated by the state random process $x_s, 0\leq s\leq t$.

The secondary role of the controller is to detect the presence of the attacker, who is able to alter the disturbance process $v_t$. We assume that the controller is able to monitor $v_t$. When an anomaly is found, the controller has the authority to trigger an alarm and halt the system's operation. We assume that the controller agent is always present.

\subsubsection{Attacker} Unlike the controller, the attacker agent may or may not be present. When the attacker is absent (or inactive), the system \eqref{SDE} is under the nominal operating condition, i.e., $v_t=w_t$. When the attacker is present, she is allowed to inject a synthetic attack signal $v_t\neq w_t$. The class of admissible attack signals $\mathcal{V}$ is specified below:
\begin{definition}
\label{def:admissible_attack}
The attack signal $v_t\in\mathbb{R}^m, 0\leq t\leq T$ is in the admissible class $\mathcal{V}$ if it is an It\^o process of the form 
\begin{equation}
\label{eq:admissible_v}
dv_t=\theta_t(\omega)dt+dw_t, \; v_0=0
\end{equation}
where $\theta_t(\omega)\in\mathcal{R}^m$ satisfies the following conditions:
\begin{itemize}
\item[(i)] $\theta: [0,T]\times\Omega \rightarrow \mathbb{R}^m$ is $\mathcal{B}[0,T]\times \mathcal{F}$ measurable where $\mathcal{B}$ is the Borel $\sigma$-algebra;
\item[(ii)] $\theta_t(\omega)$ is $\mathcal{F}_t$-adapted; and
\item[(iii)] $P(\omega\in\Omega: \int_0^t |\theta_s(\omega)|ds<\infty \; \forall t \in [0,T])=1$.
\end{itemize}
\end{definition}
\begin{remark}
The admissible class $\mathcal{V}$ of attack signals allows the attacker to inject a time-varying, possibly randomized, bias term $\theta_t(\omega)\in\mathcal{R}^m$ to the disturbance input. The progressive measurability conditions (i) and (ii), and the $L_1$ integrability condition (iii) ensure that \eqref{eq:admissible_v} is a well-defined It\^o process. 
\end{remark}

The objective of the attacker is twofold.
First, she tries to maximize the expected cost $\mathbb{E}\left[\int_0^T c_t(x_t, u_t)dt\right]$ of the system operation. This can be achieved by altering the statistics of $v_t$ from those of the natural disturbance $w_t$. 
However, a significant change in the disturbance statistics will increase the chance of being detected. Therefore, the second objective of the attacker is to stay ``stealthy" by choosing $v_t$ that is statistically similar to $w_t$.
\begin{remark}
One can consider more general attack signals of the form
\begin{equation}
\label{eq:not_admissible_v}
dv_t=\theta_t(\omega)dt+\sigma_t(\omega)dw_t, \; v_0=0
\end{equation}
with a diffusion coefficient $\sigma_t(\omega)\neq 1$.
However, from the attacker's perspective, there's no benefit in such a generalization because the choice $\sigma_t(\omega)\neq 1$ only ``hurts" the stealthiness of the attack signal.
To see this, it is sufficient to consider the following two hypotheses: 
\begin{align*}
&H_0:  dv_t=dw_t \text{ (No attack)} \\
&H_1:  dv_t=\sigma dw_t, \sigma\neq 1 \text{ (Attack on the diffusion coefficient)} 
\end{align*}
and whether the detector can tell which model has generated an observed continuous sample path $v\in\mathcal{C}[0,T]$. It can be shown that there exists a hypothesis test $\phi: \mathcal{C}[0,T] \rightarrow \{H_0, H_1\}$ that returns a correct result with probability one (See Appendix). Therefore, we can restrict ourselves to the admissible class $\mathcal{V}$ we defined as in Definition~\ref{def:admissible_attack} without loss of generality.
\end{remark}
}

\begin{figure}[!tbp]
\centering
\includegraphics[scale=0.43]{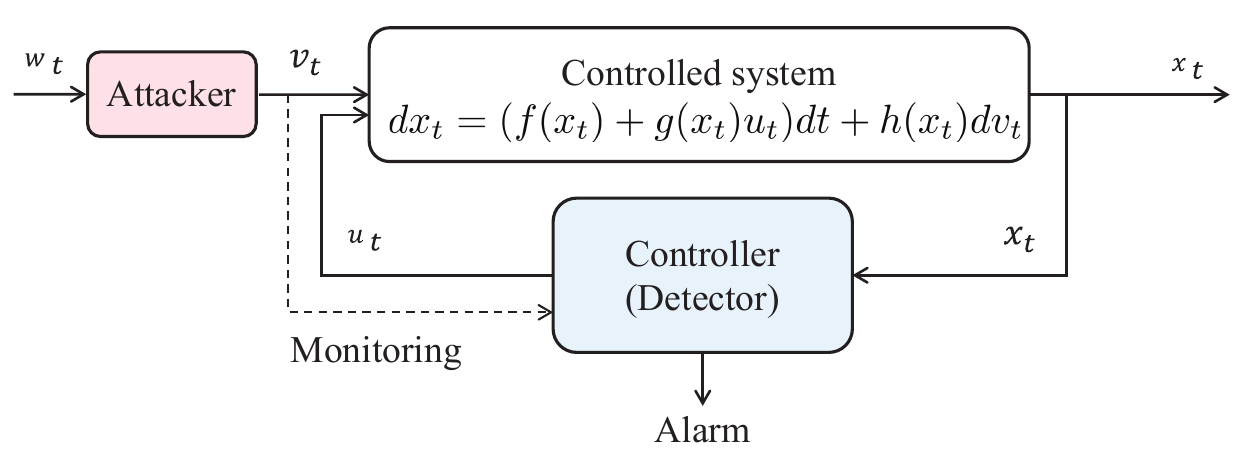}
% \vspace{-5pt}
\caption{Attacker vs controller/detector.}
\vspace{-5pt}
\label{fig:cps_security}
\end{figure}

\subsection{KL Divergence as a Stealthiness Measure}

{\color{black}
By controlling the drift term $\theta_t(\omega)$ in \eqref{eq:admissible_v}, the attacker can alter the distribution of the attack signal $v_t$.
To deal with the distributions in the space of continuous time paths, it is convenient to discuss it from the perspective of the change of measures. 
Even though $v_t$ in \eqref{eq:admissible_v} is not the standard Brownian motion under $P$ (unless $\theta_t(\omega)\equiv 0$), by Girsanov's theorem \cite{oksendal2013stochastic}, there always exists an alternative measure $Q$ in which $v_t$ is the standard Brownian motion. Girsanov's theorem also guarantees that such a measure $Q$ is equivalent to $P$ and hence the Radon-Nikodym derivative $\frac{dP}{dQ}(\omega)$ exists.

If an attack policy is fixed, then the controller's task is to determine whether an observed signal $v_t$ is generated by the natural disturbance, i.e.,
\begin{equation}
H_0: dv_t=dw_t
\end{equation}
or it is a synthetic attack signal, i.e.,
\begin{equation}
H_1: dv_t=\theta_t(\omega)dt+dw_t.
\end{equation}
This is a binary hypothesis testing problem.
Let $A\in \mathcal{F}$ be an event in which the controller triggers the alarm.
The quality of a hypothesis test (a choice of $A$) is usually evaluated in terms of the probability of type-I error (false alarm) and the probability of type-II error (failure of detection). Using notations introduced so far, the probability of type-I error can be expressed as $Q(A)$, whereas the probability of type-II error can be expressed as $P(A^c)$. 

It is well known from the Neyman-Pearson Lemma that the optimal trade-off between the type-I and type-II errors can be attained by a threshold-based likelihood ratio test. That is, for a fixed type-I error probability, the region $A$ that minimizes the type-II error probability is given by 
\begin{equation}
A=\left\{\omega\in\Omega: \log\left(\frac{dP}{dQ}(\omega)\right)\geq \tau\right\}
\end{equation}
where $\tau$ is an appropriately chosen threshold \cite{cvitanic2001generalized}.

Assuming that the controller adopts the Neyman-Pearson test (which is the most pessimistic assumption for the attacker), the attacker tries to decide on an attack policy that is difficult to detect. 
Unfortunately, for each attack policy $v_t$, it is generally difficult to compute type-I and type-II errors attainable by the Neyman-Pearson test analytically.
This poses a significant challenge in studying the attacker's policy.
Therefore, in this paper, we adopt an information-theoretic surrogate function that is more amenable to optimization -- namely, the KL divergence $D(P\|Q)=\mathbb{E}^P\log\frac{dP}{dQ}$ -- that serves as a ``stealthiness" measure.

The KL divergence has been adopted as a stealthiness measure in prior studies (e.g., \cite{bai2017data}), which can be justified by the following arguments:
The first argument is based on non-asymptotic inequalities (i.e., the inequalities that hold for any finite horizon lengths $T$). 
It follows from Pinsker's inequality and Bretagnolle-Huber inequality \cite[Theorem 14.2]{lattimore2020bandit} that
\begin{align}
Q(A)+P(A^c)&\geq 1-\sqrt{\frac{1}{2}D(P\|Q)} \label{eq:pinsker} \\
Q(A)+P(A^c)&\geq \frac{1}{2}\exp\left(-D(P\|Q)\right).
\end{align} 
These inequalities suggest that the attacker can enhance her stealthiness by choosing a policy that makes $D(P\|Q)$ small -- see Figure~\ref{fig:error_bounds}.
\begin{figure}[h]
\centering
\includegraphics[width=0.8\columnwidth]{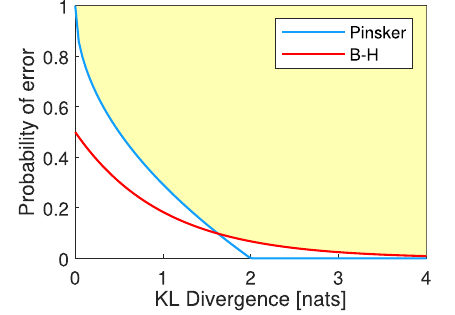}
% \vspace{-5pt}
\caption{Error bounds and achievable region (shown in yellow).}
\vspace{-5pt}
\label{fig:error_bounds}
\end{figure}\par
The second argument is asymptotic in the sense that it is concerned with how type-I and type-II errors behave as $T\rightarrow \infty$. 
The Chernoff-Stein lemma \cite{cover1999elements} states that when the type-I error probability is constrained, the exponent of the type-II error under the Neyman-Pearson test behaves as
\begin{equation}
\label{eq:error_asymptotes}
\begin{cases}
\max_\tau & -\log P(A^c) \\
\text{s.t. } & Q(A)\leq \alpha
\end{cases}=D(P\|Q) T+\mathcal{O}(\sqrt{T}).
\end{equation}
These results indicate that choosing a small $D(P\|Q)$ reduces the rate at which $P(A^c)$ decays to zero.\footnote{While \cite[Theorem 11.8.3]{cover1999elements} is specialized to i.i.d. samples, its generalization to ergodic samples \cite{polyanskiy2014lecture} has been used in \cite{bai2017data}. The continuous time result \eqref{eq:error_asymptotes} follows from \cite{tanaka2024covert}.}

\begin{remark}
Despite these justifications, the higher-order analysis of \eqref{eq:error_asymptotes} (e.g., \cite{lungu2024optimal}) suggests that the KL divergence $D(P\|Q)$ may not be an appropriate measure of stealthiness in the regime of finite $T$. Developments of stealthiness measures that outperform $D(P\|Q)$ in the finite data-length regime are critical for some applications (such as the fastest anomaly detection) and are postponed as an important future work.
\end{remark}
}

\subsection{Problem Formulation}
{\color{black}
Summarizing the discussion so far, this subsection formally states the problem studied in this paper.
\subsubsection{Stealthy Attack Synthesis}\label{sec: Worst-Case Attack Synthesis}
We first formulate a problem for optimal attack policy synthesis under a simplifying assumption that the controller's policy $u_t$ is fixed and is known to the attacker. Adopting $D(P\|Q)$ as the stealthiness measure, the attacker is incentivized to keep $D(P\|Q)$ small while maximizing the expected cost $\mathbb{E}\left[\int_0^T c_t(x_t, u_t)dt\right]$.
Introducing a constant $\lambda>0$ balancing these two requirements, the problem is formulated as follows:
\begin{problem}[KL control problem]\label{prob: KL}
\begin{equation}
\begin{aligned}
\label{eq:kld_soft}
\max_{v \in \mathcal{V}}\; &  \mathbb{E}^P\left[\int_0^T c_t(x_t,u_t)dt\right] - \lambda D(P \| Q)\\
\text{s.t.} \; & dx_t\!=\!f_t(x_t) dt +g_t(x_t)u_tdt+h_t(x_t)dv_t.
\end{aligned}
\end{equation}
\end{problem}
\vspace{2mm}
The term ``KL control problem" is inherited from \cite{theodorou2012relative} where problems with a similar structure were studied.
}

\vspace{3mm}
\subsubsection{Attack Risk Mitigation}\label{Sec: Attack Risk Mitigation}
{\color{black} Here, we consider the scenario shown in Figure~\ref{fig:cps_security} where the controller is now able to apply a control input $u_t$ to combat with the potential attack input $v_t$. As before, we assume the existence of an attack detector to detect the stealthy attack. We model the competition between the controller and the attacker as a minimax game in which the controller acts as the minimizer and the attacker acts as the maximizer. This results in the following \emph{minimax KL control} problem, which is a two-player dynamic zero-sum game between the controller (cost minimizer) and the attacker (cost maximizer):

\begin{problem}[Minimax KL control problem]\label{prob: minimax_KL}
 \begin{equation}
\label{eq:minmax_kl}
\begin{aligned}
\min_{u\in\mathcal{U}} & \max_{v\in\mathcal{V}} \mathbb{E}^{P} \left[ \int_0^T c_t(x_t, u_t)dt \right]-\lambda D(P \| Q)\\
\text{s.t.} \; & dx_t\!=\!f_t(x_t) dt\! +g_t(x_t)u_tdt+h_t(x_t)dv_t
\end{aligned}
\end{equation}   
We assume the constant $\lambda>0$ is known to both players in advance.
\end{problem}
}

\section{Stealthy Attack Synthesis via Path Integral Approach}
\label{sec:deception}
{\color{black}
This section summarizes the main technical results concerning Problem 1.
We present structural properties of the optimal attack signal $v_t$ (Theorem~\ref{thm:prob1_1}) and show how to numerically compute $v_t$ for real-time implementations.

Our first result states that Problem 1 reduces to a quadratic-cost stochastic optimal control (SOC) problem in terms of the bias input $\theta_t$:
\begin{problem}[Quadratic-cost SOC problem]\label{prob: SOC}
\begin{align}
\max_\theta\; &  \mathbb{E}^P\left[\int_0^T \left\{c_t(x_t,u_t)-\frac{\lambda}{2}\|\theta_t\|^2\right\}dt\right] \label{eq:soc_theta}\\
\text{s.t.} \; & dx_t\!=\!f_t(x_t) dt +g_t(x_t)u_tdt+h_t(x_t)(\theta_t dt + dw_t) \nonumber \\
&  \text{where $\theta_t$ satisfies conditions (i)-(iii) in Definition~\ref{def:admissible_attack}}. \nonumber
\end{align}
\end{problem}
\begin{theorem}
\label{thm:prob1_1}
Let $u_t$ be a given $\mathcal{F}_t^x$-adapted process. Then, the following statements hold:
\begin{itemize}
\item[(i)] Problem 1 is equivalent to the quadratic-cost SOC problem \eqref{eq:soc_theta}.
\item[(ii)] There exists a deterministic state feedback policy (which can be written as $\theta_t(x_t)$) for Problem \ref{prob: SOC}.
\item[(iii)] An optimal attack policy $v_t$ (governed by It\^o process  \eqref{eq:admissible_v}) is characterized by 
\begin{equation}
\label{eq:theta_opt}
\theta_t^*(x_t)=\frac{1}{\lambda}h_t^\top\partial_xV_t(x_t)
\end{equation}
where $V_t$ solves the Hamilton-Jacobi Bellman equation
\begin{equation}
\begin{aligned}\label{eq:HJB in V}
\partial_t V_t = & -\frac{1}{2\lambda}\left(\partial_x{V}_t\right)^\top h_th_t^\top\partial_x{V}_t-c_t\\ &-(f_t+g_tu_t)^\top\partial_x{V}_t-\frac{1}{2}\text{Tr}\left(h_th_t^\top\partial^2_x{V}_t\right).
\end{aligned}
\end{equation}
\end{itemize}
\end{theorem}
\begin{remark}
Theorem~\ref{thm:prob1_1} implies that the optimal attack signal can be written as $dv_t=\theta_t(x_t)dt+dw_t$, i.e., an addition of a deterministic bias $\theta_t(x_t)$ computed by \eqref{eq:theta_opt} to the natural disturbance $w_t$.
This is in contrast to the result in discrete-time \cite{bai2017data}, where the optimal attack is a randomized policy.
\end{remark}
\begin{proof}
% Statement (i) can be proved using the Girsanov theorem \cite{oksendal2003stochastic}. 
% Notice that in Problem \ref{prob: KL}, the KL divergence term $D(P\|Q)=0$ if $P\equiv Q$ i.e. $\theta_t\equiv 0$. Therefore, the KL divergence term can be thought of as the control effort. 
Notice that the KL divergence $D(P \|Q)$ term in Problem \ref{prob: KL} can be rewritten as

\begin{equation*}
  D(P \|Q)= \mathbb{E}^P \log \frac{dP}{dQ} 
\end{equation*}
Recall that $P$ is the measure in which $w_t$ is a standard Brownian motion and $Q$ is the measure in which the attack signal $v_t$ (governed by It\^o process  \eqref{eq:admissible_v}) is a standard Brownian motion. Using the Girsanov theorem \cite{oksendal2003stochastic}, we obtain
\begin{subequations}
    \label{eq:KL_girsanov}
    \begin{align}
D(P \|Q)= &\mathbb{E}^P \log \frac{dP}{dQ} \nonumber\\
=&\mathbb{E}^P \left[ \int_0^T \theta_t^\top(\omega) dw_t + \frac{1}{2}\int_0^T \|\theta_t(\omega)\|^2 dt\right] \label{eq:KL_girsanov-a}\\
=&\frac{1}{2}\mathbb{E}^P \left[ \int_0^T \|\theta_t(\omega)\|^2 dt \right]. \label{eq:KL_girsanov-b}
\end{align}
\end{subequations}
Notice that the term $\int_0^T \theta_t(\omega)^\top dw_t$ in \eqref{eq:KL_girsanov-a} is an It\^o integral. We obtain \eqref{eq:KL_girsanov-b} using the following property of It\^o integral \cite[Chapter 3]{oksendal2013stochastic}:
\begin{equation*}\label{ito integral}
    \mathbb{E}^P \int_0^T \theta_t^\top(\omega) dw_t=0.
\end{equation*} 
Equation \eqref{eq:KL_girsanov} proves statement (i). In order to prove statements (ii) and (iii), we use the \emph{dynamic programming principle}. We introduce the value function $V_t(x_t)$ for each time $t\in[0,T)$ and the state $x_t\in\mathbb{R}^n$ as 
\begin{equation}
\label{eq:value_f}
V_t(x_t):=\max_\theta  \mathbb{E}^P\int_t^T \left(c_t(x_s,u_s) - \frac{\lambda}{2}\|\theta_s\|^2\right)ds.
\end{equation} 
 Let us define \begin{equation}
\label{eq:value_f_min}
\overline{V_t}(x_t):=\min_\theta  \mathbb{E}^P\int_t^T \left(-c_s (x_s,u_s) + \frac{\lambda}{2}\|\theta_s\|^2\right)ds.
\end{equation} 
Note that $V_t(x_t)=-\overline{V_t}(x_t)$. The stochastic Hamilton-Jacobi-Bellman (HJB) equation \cite{fleming2006controlled, stengel1994optimal} associated with \eqref{eq:value_f_min} is expressed as follows:
\begin{equation}\label{eq:HJB}
    \begin{aligned}
  -\partial_t\overline{V}_t =  \min_{{\theta}_t}\Big[&\frac{\lambda}{2}\|\theta_t\|^2+\left(f_t+g_tu_t + h_t\theta_t\right)^{\top}\partial_x\overline{V}_t\\
  & -c_t+\frac{1}{2}\text{Tr}\left(h_t h_t^\top\partial^2_x\overline{V}_t\right)\Big].
\end{aligned}
\end{equation}
Solving \eqref{eq:HJB}, we get the optimal $\theta_t$ as
\begin{equation}\label{eq:theta_star}
    \theta_t^*(x_t) = -\frac{1}{\lambda}h_t^\top\partial_x\overline{V_t}(x_t).
\end{equation}
Putting \eqref{eq:theta_star} in \eqref{eq:HJB}, we get
\begin{equation}\label{eq:HJB2}
\begin{aligned}
  -\partial_t\overline{V}_t = & -\frac{1}{2\lambda}\left(\partial_x\overline{V}_t\right)^\top h_th_t^\top\partial_x\overline{V}_t-c_t\\ &+(f_t+g_tu_t)^\top\partial_x\overline{V}_t+\frac{1}{2}\text{Tr}\left(h_th_t^\top\partial^2_x\overline{V}_t\right).
\end{aligned}
\end{equation}
Since $V_t(x_t)=-\overline{V_t}(x_t)$, this proves statements (ii) and (iii).
% (ii)(iii) Maybe refer to existing results?
\end{proof}
Generally, it is challenging to compute \eqref{eq:theta_opt} since it requires the solution $V_t(x_t)$ to a nonlinear, possibly high-dimensional, partial differential equation (PDE) \eqref{eq:HJB in V}.
Fortunately, the structure of the optimal control problem \eqref{eq:soc_theta} allows for an application of the \emph{path integral method} \cite{kappen2005path},  offering a Monte-Carlo-based attack signal synthesis. The applicability of the path-integral method to the quadratic-cost SOC problems shown in Problem \ref{prob: SOC} was first pointed out by the work of Kappen\cite{kappen2005path}.
For each time $t\in[0,T)$ and the state $x_t\in\mathbb{R}^n$, the path-integral method allows the adversary to compute the optimal attack signal $\theta_t(x_t)$ by evaluating the path integrals along randomly generated trajectories $\{x_s, u_s\}$, $t\leq s \leq T$ starting from $x_t$. The next result provides the details.
% The next result is the instrumental for developing a numerical algorithm.
}

\begin{theorem}\label{Theorem: sol of SOC}
  The solution of \eqref{eq:value_f} exists, is unique and is given by \begin{equation}
\label{eq:value_f2}
V_t(x_t)=\lambda \log \mathbb{E}^Q \left[\exp \left\{\frac{1}{\lambda}\int_t^T c_s(x_s, u_s)ds \right\}\right]. 
\end{equation}  
Furthermore, the optimal bias input $\theta_t^*(x_t)$ is given by
\begin{equation}\label{eq: theta_star} \!\!\!\!\!\theta_t^*dt\!=\!\mathcal{H}_t(x_t)\frac{\mathbb{E}^Q\left[\text{exp}{\left\{\frac{1}{\lambda}\int_t^T\!\! c_s(x_s, u_s)ds\right\}}h_t(x_t)d{w}_t\right]}{\mathbb{E}^Q\left[\text{exp}{\left\{\frac{1}{\lambda}\int_t^T \!\!c_s(x_s, u_s)ds\right\}}\right]} 
\end{equation}
where the matrix $\mathcal{H}_t(x_t)$ is defined as
\begin{equation*}
    \mathcal{H}_t(x_t) = h_t^\top(x_t)\left(h_t(x_t)h_t(x_t)^\top\right)^{-1}.
\end{equation*}
\end{theorem}
\vspace{2mm}
\begin{proof}
 In Theorem \ref{thm:prob1_1}, we proved that the value function $\overline{V}_t(x_t)$ \eqref{eq:value_f_min} satisfies the PDE \eqref{eq:HJB2}. We introduce the exponential transformation of the value function $\overline{V}_t(x_t) = -\lambda\log\Psi_t(x_t)$ (known as Cole-Hopf transformation \cite{kappen2005path}). This will reformulate \eqref{eq:HJB2} as:
\begin{equation}\label{eq:Psi}
    \!\partial_t\Psi_t\!=\!\frac{-c_t\Psi_t}{\lambda}\!-\!(f_t+g_tu_t))^\top\partial_x\Psi_t-\frac{1}{2}\text{Tr}\left(h_th_t^\top\partial^2_x\Psi\right).
\end{equation}
The PDE \eqref{eq:Psi} is linear in terms of $\Psi_t$ and is known as the backward Chapman-Kolmogorov PDE \cite{williams2017model}. According to the Feynman-Kac lemma \cite{oksendal2003stochastic}, the solution of the linear PDE \eqref{eq:Psi} exists and is unique in the sense that $\Psi_t$ solving \eqref{eq:Psi} is given by
\begin{equation}\label{eq:Psi_sol}
    \Psi_t(x_t) = \mathbb{E}^Q \left[\exp \left\{\frac{1}{\lambda}\int_t^T c_s(x_s, u_s)ds \right\}\right]
\end{equation}
where $Q$ is the probability measure in which the attack $v_t$ is a standard Brownian motion i.e., $\theta_s$, $t\leq s\leq T$ is zero. Since $V_t(x_t) = -\overline{V_t}(x_t) = \lambda\log\Psi_t(x_t)$, using \eqref{eq:Psi_sol}, we get the desired result \eqref{eq:value_f2}. Solving \eqref{eq:theta_opt}, i.e., taking the gradient of ${V_t(x_t)}$ with respect to $x$, we get the optimal bias input \eqref{eq: theta_star}.
\end{proof}
Since the right-hand side of \eqref{eq:value_f2} contains the expectation operation with respect to $Q$, using the strong law of large numbers \cite{durrett2019probability}, we can prove that as $N\rightarrow\infty$,
\[
\lambda \log \left[\frac{1}{N}\sum_{i=1}^N \exp \left\{\frac{1}{\lambda}\int_t^T c_s(x^i_s, u^i_s)ds\right\}\right] \overset{a.s.} {\rightarrow} V_t(x_t)
\]
where $\{x_s^i, u^i_s, t\leq s\leq T\}_{i=1}^N$ are randomly drawn sample paths from distribution $Q$.
Since $Q$ is the measure in which $v_t$ is the standard Brownian motion, generating such a sample ensemble is easy. It suffices to perform $N$ independent simulations of the dynamics $dx_s=f_s(x_s)ds +g_s(x_s)u_sds+ h_s(x_s)dw_s$. 
% Since this Monte-Carlo-based solution is less susceptible to the \emph{curse of dimensionality} than alternative methods (e.g., finite difference method \cite{patil2022chance}), the path integral method is considered beneficial, especially for high-dimensional applications \cite{theodorou2010reinforcement}.\par
Similarly, the optimal bias input $\theta_t^*$ \eqref{eq: theta_star} can be readily computed by the same simulated ensemble $\{x_s^i, u_s^i, t\leq s\leq T\}_{i=1}^N$ and their path costs \cite{kappen2005path, williams2016aggressive}. According to the strong law of large numbers, as $N\rightarrow\infty$,

\begin{equation} \label{eq: theta_star_MC}\!\!\mathcal{H}_t(x_t)\frac{\frac{1}{N}\sum_{i=1}^N \text{exp}{\left\{\frac{1}{\lambda}\int_t^T c_s(x_s^{i}, u_s^{i})ds\right\}}h_t(x_t)\epsilon}{\frac{\sqrt{\Delta t}}{N}\sum_{i=1}^N \text{exp}{\left\{\frac{1}{\lambda}\int_t^T c_s(x_s^i, u_s^i)ds\right\}}} \overset{a.s.} {\rightarrow} \theta^*_t
\end{equation}
where $\epsilon\sim\mathcal{N}(0,1)$ and $\Delta t$ is the step size. Equation \eqref{eq: theta_star_MC} implies that if the attacker has a simulator engine that can generate a large number of sample trajectories $\{x_t^i, u_t^i\}_{i=1}^N$ from the distribution $Q$ starting from the current state-time pair $(x_t,t)$, then drift term $\theta_t^*$ \eqref{eq: theta_star_MC} of the optimal attack signal $v_t^*$  can be computed directly from the sample ensemble $\{x_t^i, u_t^i\}_{i=1}^N$. A notable feature of such a simulator-driven attack synthesis is that it computes the optimal stealthy attack signal in real time for nonlinear CPS without requiring an explicit model of the system or an explicit policy synthesis step. 

\section{Attack Risk Mitigation}\label{sec: Attack impact mitigation}
\label{sec:mitigation}

In this section, we turn our attention to the controller's problem, who is interested in mitigating the risk of stealthy attacks when the attacker is present. In Section \ref{sec: Connections with Risk-Sensitive Control and Two-Player Zero-Sum Stochastic Differential Game}, we establish a connection between the minimax KL control problem (Problem \ref{prob: minimax_KL}) with risk-sensitive control and two-player zero-sum stochastic differential game. In Sections \ref{sec: Risk-Sensitive Control via Path Integral Approach} and \ref{sec: Two-player Zero-Sum Stochastic Differential Game via Path Integral Approach}, we present the solutions of risk-sensitive control problems and two-player stochastic differential games, 
respectively, using the path integral approach. This allows the controller to synthesize the risk-mitigating control signals via Monte Carlo simulations.  

\subsection{Connections with Risk-Sensitive Control and Two-Player Zero-Sum Stochastic Differential Game}\label{sec: Connections with Risk-Sensitive Control and Two-Player Zero-Sum Stochastic Differential Game}
\begin{figure*}[h]
    \centering   \includegraphics[scale=0.37]{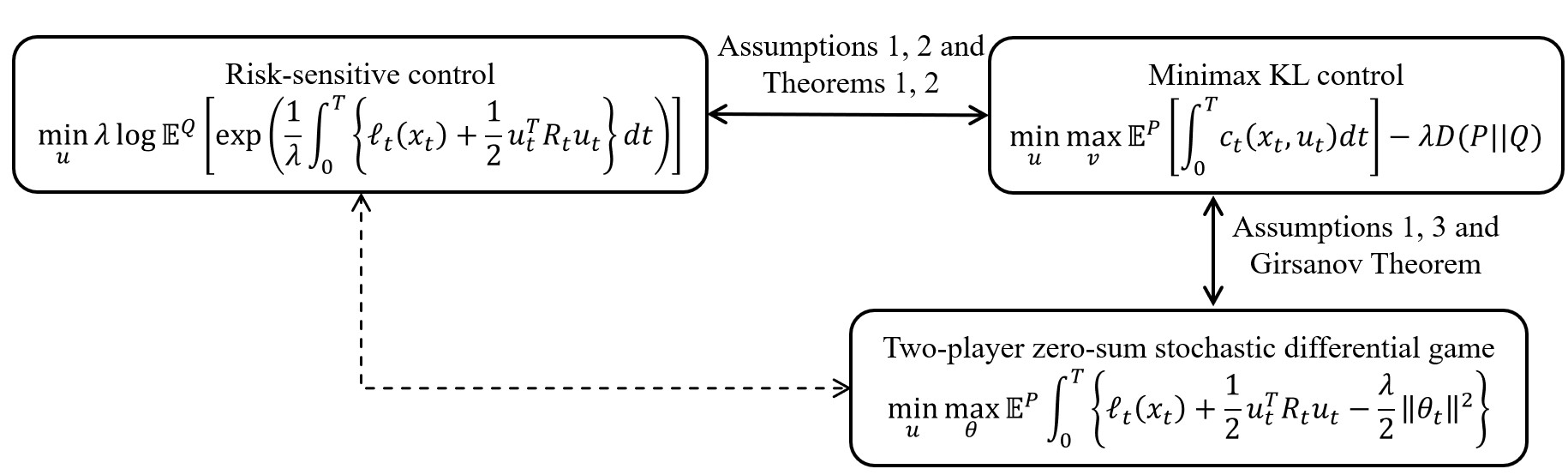}
    \vspace{-1ex}
    \caption{Connections between minimax KL control, risk-sensitive control, and two-player stochastic differential game}
    \label{fig:jacobson}
\end{figure*}
First, we establish the connection between the minimax KL control problem \eqref{eq:minmax_kl} and the \emph{risk-sensitive control problem} \cite{fleming2006controlled, whittle1981risk, jacobson1973optimal}. To this end, we make the following assumption:
\begin{assumption}\label{Assum:quadratic cost}
 the  cost function $c_t$ is quadratic in $u_t$:
\begin{equation}\label{eq: ct}
c_t(x_t, u_t)=\ell_t(x_t)+\frac{1}{2}u_t^\top R_t(x_t) u_t
\end{equation}
where $R_t(x_t)\succeq 0$ for all $t$. 
\end{assumption}
% Now, since Girsanov's theorem (Theorem \ref{thm: Girsanov}) implies that $D(P\|Q)=\frac{1}{2}\mathbb{E}^Q \int_0^T \|\theta_t\|^2 dt$, the KL divergence term $D(P \| Q)$ in \eqref{eq:minmax_kl} does not depend on the minimizer's policy $u_t$. 
We can determine the value of the inner maximization problem in \eqref{eq:minmax_kl} using \eqref{eq:value_f2}. Using \eqref{eq:value_f2} and \eqref{eq: ct}, Problem \ref{prob: minimax_KL} can be written equivalently as follows:

\begin{problem}[Risk-sensitive control problem]\label{prob: risk-sensitive control}
 \begin{equation}
\label{eq:risk_sensitive}
\begin{aligned}
\min_{u} \; & \lambda \log \mathbb{E}^Q \!\left[\exp \!\left(\frac{1}{\lambda}\!\int_0^T\!\! \left\{\!\ell_t(x_t)\!+\!\frac{1}{2}u_t^{\!\top} R_t u_t\!\right\}dt\!\right)\!\right]\\
\text{s.t.} \; & dx_t\!=\!f_t(x_t) dt\! +g_t(x_t)u_tdt+h_t(x_t)dw_t.
\end{aligned}
\end{equation}   
\end{problem}
\vspace{3mm}
 The objective function in \eqref{eq:risk_sensitive} is also related to the risk measure called \emph{Entropic value-at-risk (EVaR)}
\cite{ahmadi2012entropic}.
This equivalence shows the intimate relationship between the minimax KL control problem and the risk-sensitive control problem. Problem \ref{prob: risk-sensitive control} is a class of risk-sensitive control problems with certain structural constraints. Specifically, the cost function is quadratic in the control input $u_t$, and the state equation is affine in both the control input $u_t$ and the noise $w_t$. This class of risk-sensitive control problems can be solved using the basic path-integral method \cite{broek2012risk}.\par

Now, we establish the connection between the minimax KL control problem \eqref{eq:minmax_kl} and the \emph{two-player zero-sum stochastic differential game}. Using the Girsanov's theorem (Theorem \ref{thm:prob1_1}-(i)) and assuming that the cost function $c_t$ is quadratic in $u_t$ \eqref{eq: ct}, Problem \ref{prob: minimax_KL} can be rewritten as

\begin{problem}[Two-player stochastic differential game]\label{prob: game}
\begin{equation}
\label{eq:h_inf}
\begin{aligned}
\min_u&\max_\theta \mathbb{E}^P\! \left[\int_0^T\!\! \left\{\ell_t(x_t)+\frac{1}{2}u_t^\top R_t u_t -\frac{\lambda}{2}\|\theta_t\|^2\right\}dt\right]\\
\text{s.t.} \; & dx_t\!=\!f_t(x_t) dt\! +g_t(x_t)u_tdt+h_t(x_t)\left(\theta_t dt \!+\!dw_t\!\right).
\end{aligned}
\end{equation}
\end{problem}
\vspace{3mm}
Problem \ref{prob: game} is a class of two-player zero-sum stochastic differential games with certain structural constraints. Specifically, the cost function is quadratic in $u_t$, $\theta_t$, and the state equation is affine in both the inputs $u_t$, $\theta_t$, and the noise input $w_t$. This class of stochastic differential games can be solved using the basic path-integral method \cite{patil2023risk}.\par

The equivalence of Problems \ref{prob: minimax_KL}, \ref{prob: risk-sensitive control}, and \ref{prob: game} under assumptions \ref{Assum:quadratic cost}, \ref{Assumption: linearity_risk} and \ref{Assumption: linearity} is shown in Figure~\ref{fig:jacobson}. Problem \ref{prob: game} is also known as the \emph{non-linear $H_\infty$ control} problem. We remark that the connection between the risk-sensitive control and the $H_\infty$ control (indicated by a dashed line in Figure~\ref{fig:jacobson}) for a class of linear systems is already known (e.g., \cite{jacobson1973optimal}). 
Figure~\ref{fig:jacobson} re-establishes this connection, possibly for a broader class of dynamical systems.
% through the lens of minimax KL control problems.

\subsection{Risk-Sensitive Control via Path Integral Approach}\label{sec: Risk-Sensitive Control via Path Integral Approach}

The applicability of the path-integral method to the class of risk-sensitive control problems shown in Problem \ref{prob: risk-sensitive control} was pointed out by the work of Broek et. al. in \cite{broek2012risk}. For each time $t\in[0,T)$ and the state $x_t\in\mathbb{R}^n$, the path-integral method, under a certain assumption, allows the controller to compute the optimal policy $u_t$ under the presence of worst-case attack by evaluating the path integrals along randomly generated trajectories $\{x_s\}$, $t\leq s \leq T$ starting from $x_t$. This can be derived from the fact that under a certain assumption, the value function of the risk-sensitive control problem \eqref{eq:risk_sensitive} can be computed by Monte-Carlo sampling. To see this, for each state-time pair $(x_t,t)$, introduce the value function 
\begin{equation}
\label{eq:value_f_risk}
\!\!V_t(x_t)\!=\min_u \lambda \log \mathbb{E}^Q \!\left[\exp \!\left(\frac{1}{\lambda}\!\int_t^T\!\! \left\{\!\ell_s(x_s)\!+\!\frac{1}{2}u_s^{\!\top} R_s u_s\!\right\}ds\!\right)\!\right]\!\!.
\end{equation}
First, we make the following assumption which is essential to solve \eqref{eq:value_f_risk} using the path integral method \cite{broek2012risk}.

\begin{assumption}\label{Assumption: linearity_risk}
   For all $(x_t,t)$, there exists a constant $0<\xi<\lambda$ satisfying the following equation: 
   \begin{equation}\label{eq: linearizability_risk}
    \!\!h_t(x_t)  h_t^{\!\top}(x_t)\!\! =\! \xi g_t(x_t)R_t^{-1}(x_t)g_t^\top(x_t).
\end{equation}

\end{assumption}
\vspace{3mm}
Assumption \ref{Assumption: linearity_risk} is similar to the assumption required in the path integral formulation of a stochastic control problem \cite{satoh2016iterative}. A possible interpretation of condition \eqref{eq: linearizability_risk} is that in a direction with high noise variance, the control cost of the risk-sensitive control problem \eqref{eq:risk_sensitive} has to be low. Therefore, the weight of the control cost $R_t$ need to be tuned appropriately for the given $\lambda$, $h_t(x_t)$ and  $g_t(x_t)$ for all $t$. See \cite{broek2012risk} for further discussion on this condition.

\begin{theorem}\label{thm: risk-sensitive control}
    Under Assumption \ref{Assumption: linearity_risk}, the solution of \eqref{eq:value_f_risk} exists, is unique and is given by
    \begin{equation}
\label{eq:value_f2_risk}
V_t(x_t)=-\gamma \log \mathbb{E}^Z \left[\exp \left\{-\frac{1}{\gamma}\int_t^T \ell_s(x_s)ds \right\}\right]
\end{equation}
where $Z$ is the probability measure defined under the system dynamics \eqref{SDE} when both $u_s$ and $\theta_s$, $t\leq s\leq T$ are zero and 
\begin{equation}\label{eq: gamma}
\gamma = \frac{\xi\lambda}{\lambda-\xi} .   
\end{equation}
Furthermore, the optimal controller signal $u_t^*$ is given by
\begin{equation}\label{eq: u_star_risk2} \!\!\!\!\!u_t^*dt\!=\!\mathcal{H}_t(x_t)\frac{\mathbb{E}^Z\left[\text{exp}{\left\{-\frac{1}{\gamma}\int_t^T \ell_s(x_s)ds\right\}}h_t(x_t)d{w}_t\right]}{\mathbb{E}^Z\left[\text{exp}{\left\{-\frac{1}{\gamma}\int_t^T \ell_s(x_s)ds\right\}}\right]} 
\end{equation}
where the matrix $\mathcal{H}_t(x_t)$ is defined as
\begin{equation*}
    \!\!\mathcal{H}_t(x_t)\!=\!R_t^{-1} g_t^\top(x_t)\!\left(\!g_t(x_t)R_t^{-1}g_t^\top(x_t)\! -\frac{1}{\lambda} h_t(x_t)h_t(x_t)^\top\!\right)^{\!-1}
\end{equation*}
\end{theorem}
\begin{proof}
   We use \textit{dynamic programming principle} and \textit{Feynman-Kac lemma} \cite{oksendal2003stochastic} to prove this theorem. The stochastic Hamilton-Jacobi-Bellman (HJB) equation \cite{broek2012risk} associated with \eqref{eq:value_f_risk} is expressed as follows:
\begin{equation}\label{eq:HJB_risk}
    \begin{aligned}
  -\partial_t{V}_t =  \min_{{u}_t}\Big[&\frac{1}{2}u_t^\top\!R_tu_t+ \ell_t+ \left(\!f_t\!+\!g_tu_t\right)^\top\!\partial_xV_t\\
  & +\frac{1}{2\lambda}\|h_t^\top\partial_xV_t\|^2 + \frac{1}{2}\text{Tr}\left(h_t h_t^\top\partial^2_x{V}_t\right)\Big]
\end{aligned}
\end{equation}
Solving \eqref{eq:HJB_risk}, we get the optimal $u_t$ as 
   \begin{equation}\label{eq:u_star_risk}
       u_t^*(x_t) = -R_t^{-1}g_t^\top\partial_xV_t(x_t).
   \end{equation}
   Putting \eqref{eq:u_star_risk} in \eqref{eq:HJB_risk}, we get
   
   \begin{equation}\label{eq:HJB_risk2}
    \begin{aligned}
         \!\!-\partial_tV_t\!=&\ell_t\! +\!f_t^\top\!\partial_xV_t\!+\!\frac{1}{2}\text{Tr}\left(h_th_t^\top\partial^2_xV_t\right)\\ 
         &+\frac{1}{2}\!\left(\partial_xV_t\right)^\top\!\!\left(\frac{1}{\lambda}h_th_t^\top\!-\!g_tR_t^{-1}g_t^\top\right)\!\partial_xV_t.
\end{aligned}
\end{equation}
Let $\Psi_t(x_t)$ be the logarithmic transformation of the value function $V_t(x_t)$ (known as Cole-Hopf transformation) defined as $V_t(x_t) = -\gamma\log\Psi_t(x_t)$
where $\gamma>0$ is a proportionality constant to be defined. Applying this transformation of the value function to \eqref{eq:HJB_risk2} yields \begin{equation}\label{eq:Psi_risk}
 \begin{aligned}         \!\!\!\!\partial_t\Psi_t\!=&\frac{\ell_t\Psi_t}{\gamma}-\!\frac{1}{2}\text{Tr}\!\left(h_th_t^\top\!\partial^2_x\Psi_t\right)\!+\!\frac{1}{2\Psi_t}\!\left(\partial_x\Psi_t\right)^{\!T}\!\!h_th_t^\top\!\partial_x\Psi_t \\
         &\!\!\!\!\!\!\!\!\!\!\!\!\!\!+\!\frac{\gamma}{2\Psi_t}\!\left(\partial_x\Psi_t\right)^\top\!\!\left(\frac{1}{\lambda}h_th_t^\top\!-\! g_tR_t^{-1}g_t^\top\right)\!\partial_x\Psi_t\!-\!f_t^\top\partial_x\Psi_t.
          \end{aligned}               
 \end{equation}
If we assume that Assumption \ref{Assumption: linearity_risk} holds and $\gamma$ satisfies \eqref{eq: gamma}, we obtain a linear PDE in $\Psi_t$ known as the backward Chapman-Kolmogorov PDE:

\begin{equation}\label{eq:Psi_risk2}
 \!\partial_t\Psi_t\!=\!\frac{\ell_t\Psi_t}{\lambda}\!-\!f_t^\top\partial_x\Psi_t-\frac{1}{2}\text{Tr}\left(h_th_t^\top\partial^2_x\Psi_t\right).
\end{equation}
According to the Feynman-Kac lemma \cite{oksendal2003stochastic}, the solution of the linear PDE \eqref{eq:Psi_risk2} exists and is unique in the sense that $\Psi_t$ solving \eqref{eq:Psi_risk2} is given by
\begin{equation}\label{eq:Psi_risk3}
    \Psi_t(x_t)= \mathbb{E}^Z \left[\exp \left\{-\frac{1}{\gamma}\int_t^T \ell_s(x_s)ds \right\}\right]
\end{equation}
   where $Z$ is the probability measure defined under the system dynamics \eqref{SDE} when both $u_s$ and $\theta_s$, $t\leq s\leq T$ are zero. Since $V_t(x_t) = -\gamma\log\Psi_t(x_t)$, using \eqref{eq:Psi_risk3}, we get the desired result \eqref{eq:value_f2_risk}. Solving \eqref{eq:u_star_risk} (taking the gradient of \eqref{eq:value_f2_risk} with respect to $x$), we get the optimal controller policy \eqref{eq: u_star_risk2}.
\end{proof}

% Instead of computing $V_t(x_t)$ by solving a backward HJB PDE \cite{broek2012risk}, the path integral method, under Assumption \ref{Assumption: linearity_risk}, attempts to approximate it by using forward Monte-Carlo simulations. 
\begin{remark}
Notice that according to Assumption \ref{Assumption: linearity_risk}, $0<\xi<\lambda$. Therefore, $\gamma$ is always positive.     
\end{remark}

Since the right-hand side of \eqref{eq:value_f2_risk} contains the expectation operation with respect to $Z$, using the strong law of large numbers \cite{durrett2019probability}, we can prove that as $N\rightarrow\infty$,
\[-\gamma \log \left[\frac{1}{N}\sum_{i=1}^N \exp \left\{-\frac{1}{\gamma}\int_t^T \ell_s(x^i_s)ds\right\}\right]\overset{a.s.} {\rightarrow} V_t(x_t)
\]
where $\{x_s^i, t\leq s\leq T\}_{i=1}^N$ are randomly drawn sample paths from distribution $Z$. Since $Z$ is the measure in which both $u_s$ and $\theta_s$, $t\leq s\leq T$ are zero, generating such a sample ensemble is easy.
It suffices to perform $N$ independent simulations of the dynamics $dx_s=f_s(x_s)ds + h_s(x_s)dw_s, \; x_t=x$.  

\subsection{Two-Player Zero-Sum Stochastic Differential Game via Path Integral Approach}\label{sec: Two-player Zero-Sum Stochastic Differential Game via Path Integral Approach}
The applicability of the path-integral method to the class of stochastic differential games shown in Problem \ref{prob: game} was pointed out by our work in \cite{patil2023risk}. For each time $t\in[0,T)$ and the state $x_t\in\mathbb{R}^n$, the path-integral method, under a certain assumption, allows both the controller and the adversary to compute the optimal policies $u_t$ and $\theta_t$- known as \emph{saddle-point policies} \cite[Chapter 2]{bacsar2008h} by evaluating the path integrals along randomly generated trajectories $\{x_s\}$, $t\leq s \leq T$ starting from $x_t$. This can be derived from the fact that under a certain assumption, the value of the game \eqref{eq:h_inf} can be computed by Monte-Carlo sampling. To see this, for each state-time pair $(x_t,t)$, introduce the value of the game as 
\begin{equation}
\label{eq:value_f_game}
V_t(x_t)\!=\!\min_u \max_\theta  \mathbb{E}^P\!\!\left[\int_t^T\!\!\! \left(\!\ell_s(x_s)\! + \!\frac{1}{2}u_s^{\!\top} R_s u_s\! -\! \frac{\lambda}{2}\|\theta_s\|^2\!\!\right)ds\right]\!.
\end{equation}
First, we make the following assumption, which is essential to solve \eqref{eq:value_f_game} using the path integral approach \cite{patil2023risk}.

\begin{assumption}\label{Assumption: linearity}
   For all $(x,t)$, there exists a constant $\alpha>0$ satisfying the following equation: 
   \begin{equation}\label{eq: linearizability_game}
    \!\!h_t(x_t)  h_t^{\!\top}(x_t)\!\! =\! \alpha\!\left(\!g_t(x_t)R_t^{-1}\!g_t^\top(x_t)\! -\! \frac{1}{\lambda} h_t(x_t)  h_t^{\!\top}(x_t)\!\!\right)\!.
\end{equation}
\end{assumption}
\vspace{3mm}
Assumption \ref{Assumption: linearity} is similar to the assumption required in the path integral formulation of a single-agent stochastic control problem \cite{satoh2016iterative}. A possible interpretation of condition \eqref{eq: linearizability_game} is that in a direction with high noise variance, the controller's control cost has to be low. Therefore, the weight of the control cost $R_t$ need to be tuned appropriately for the given $\lambda$, $h_t(x_t)$ and  $g_t(x_t)$ for all $t$. See \cite{patil2023risk} for further discussion on this condition.

% Instead of computing $V_t(x_t)$ by solving a backward Hamilton-Jacobi-Isaacs (HJI) PDE \cite{patil2023risk}, the path integral method, under Assumption \ref{Assumption: linearity}, attempts to approximate it by using forward Monte-Carlo simulations. 
\begin{theorem}\label{thm: two-player game}
  Under Assumption \ref{Assumption: linearity}, the solution of \eqref{eq:value_f_game} exists, is unique and is given by
  \begin{equation}
\label{eq:value_f2_game}
V_t(x_t)=-\alpha \log \mathbb{E}^Z \left[\exp \left\{-\frac{1}{\alpha}\int_t^T \ell_s(x_s)ds \right\}\right]
\end{equation}
where $Z$ is the probability measure defined under the system dynamics \eqref{SDE} when both $u_s$ and $\theta_s$, $t\leq s\leq T$ are zero.
Furthermore, the saddle-point policies of Problem \ref{prob: game} are given by 
\begin{equation}\label{eq: u_star_game} \!\!\!\!\!u_t^*dt\!=\!\mathcal{H}^u_t(x_t)\frac{\mathbb{E}^Z\left[\text{exp}{\left\{-\frac{1}{\alpha}\int_t^T \ell_s(x_s)ds\right\}}h_t(x_t)d{w}_t\right]}{\mathbb{E}^Z\left[\text{exp}{\left\{-\frac{1}{\alpha}\int_t^T \ell_s(x_s)ds\right\}}\right]} 
\end{equation}
where the matrix $\mathcal{H}^u_t(x_t)$ is defined as
\begin{equation*}
    \!\!\mathcal{H}^u_t(x_t)\!=\!R_t^{-1} g_t^\top(x_t)\!\left(\!g_t(x_t)R_t^{-1}g_t^\top(x_t)\! -\frac{1}{\lambda} h_t(x_t)h_t(x_t)^\top\!\right)^{\!-1}
\end{equation*}
and
\begin{equation}\label{eq: theta_star_game} \!\!\!\!\!\theta_t^*dt\!=\!\mathcal{H}^\theta_t(x_t)\frac{\mathbb{E}^Z\left[\text{exp}{\left\{-\frac{1}{\alpha}\int_t^T \ell_s(x_s)ds\right\}}h_t(x_t)d{w}_t\right]}{\mathbb{E}^Z\left[\text{exp}{\left\{-\frac{1}{\alpha}\int_t^T \ell_s(x_s)ds\right\}}\right]} 
\end{equation}
where the matrix $\mathcal{H}^\theta_t(x_t)$ is defined as
\begin{equation*}
    \!\!\mathcal{H}^\theta_t(x_t)\!=\!-\frac{1}{\lambda} h_t^\top(x_t)\!\left(\!g_t(x_t)R_t^{-1}g_t^\top(x_t)\! -\frac{1}{\lambda} h_t(x_t)h_t(x_t)^\top\!\right)^{\!\!-1}\!\!\!\!.
\end{equation*}
\end{theorem}
\begin{proof}
   The stochastic Hamilton-Jacobi-Isaacs (HJI) equation \cite{bacsar2008h} associated with \eqref{eq:value_f_game} is expressed as follows:
   \begin{equation}\label{eq:HJI}
    \begin{aligned}
         \!\!-\partial_tV_t\!=\!\min_{{u_t}} \max_{{\theta_t}}\!\Bigg[&\frac{1}{2}u_t^\top\!R_tu_t-\frac{\lambda}{2}\|\theta_t\|^2+\!\ell_t\!\\
         &\!\!\!\!\!\!\!\!\!\!\!\!\!\!\!\!\!\!\!\!\!\!\!\!+\!\left(\!f_t\!+\!g_tu_t\!+\!h_t\theta_t\right)^\top\!\partial_xV_t+\!\frac{1}{2}\text{Tr}\!\left(h_th_t^\top\partial^2_xV_t\right)\!\Bigg]\!.
          \end{aligned} 
\end{equation}
   Solving \eqref{eq:HJI}, we get 
   \begin{equation}\label{eq:saddle-point}
       u_t^*(x_t) \!= \!-R_t^{-1}g_t^\top\partial_xV_t(x_t), \quad \theta_t^*(x_t)\! =\!\frac{1}{\lambda}h_t^\top\partial_xV_t(x_t).
   \end{equation}
Putting \eqref{eq:saddle-point} in \eqref{eq:HJI}, we get  
\begin{equation}\label{eq:HJI2}
    \begin{aligned}
         \!\!-\partial_tV_t\!=&\ell_t\! +\!f_t^\top\!\partial_xV_t\!+\!\frac{1}{2}\text{Tr}\left(h_th_t^\top\partial^2_xV_t\right)\\ 
         &+\frac{1}{2}\!\left(\partial_xV_t\right)^\top\!\!\left(\frac{1}{\lambda}h_th_t^\top\!-\!g_tR_t^{-1}g_t^\top\right)\!\partial_xV_t.
\end{aligned}
\end{equation}
Let $\Psi_t(x_t)$ be the logarithmic transformation of the value function $V_t(x_t)$ defined as $V_t(x_t) = -\alpha\log\Psi_t(x_t)$
where $\alpha>0$ is a proportionality constant defined by the Assumption \ref{Assumption: linearity}. Applying this transformation of the value function to \eqref{eq:HJI2} yields \begin{equation}\label{eq:Psi_game}
 \begin{aligned}         \!\!\!\!\partial_t\Psi_t\!=&\frac{\ell_t\Psi_t}{\alpha}-\!\frac{1}{2}\text{Tr}\!\left(h_th_t^\top\!\partial^2_x\Psi_t\right)\!+\!\frac{1}{2\Psi_t}\!\left(\partial_x\Psi_t\right)^{\!T}\!\!h_th_t^\top\!\partial_x\Psi_t \\
         &\!\!\!\!\!\!\!\!\!\!\!\!\!\!+\!\frac{\alpha}{2\Psi_t}\!\left(\partial_x\Psi_t\right)^\top\!\!\left(\frac{1}{\lambda}h_th_t^\top\!-\! g_tR_t^{-1}g_t^\top\right)\!\partial_x\Psi_t\!-\!f_t^\top\partial_x\Psi_t.
          \end{aligned}               
 \end{equation}
By assuming an $\alpha$ satisfying Assumption \ref{Assumption: linearity} holds in \eqref{eq:Psi_game}, we obtain a linear PDE in $\Psi_t$ known as the backward Chapman-Kolmogorov PDE:

\begin{equation}\label{eq:Psi_game2}
 \!\partial_t\Psi_t\!=\!\frac{\ell_t\Psi_t}{\lambda}\!-\!f_t^\top\partial_x\Psi_t-\frac{1}{2}\text{Tr}\left(h_th_t^\top\partial^2_x\Psi_t\right).
\end{equation}
According to the Feynman-Kac lemma \cite{oksendal2003stochastic}, the solution of the linear PDE \eqref{eq:Psi_game2} exists and is unique in the sense that $\Psi_t$ solving \eqref{eq:Psi_game2} is given by
\begin{equation}\label{eq:Psi_game3}
    \Psi_t(x_t)= \mathbb{E}^Z \left[\exp \left\{-\frac{1}{\alpha}\int_t^T \ell_s(x_s)ds \right\}\right]
\end{equation}
   where $Z$ is the probability measure defined under the system dynamics \eqref{SDE} when both $u_s$ and $\theta_s$, $t\leq s\leq T$ are zero. Since $V_t(x_t) = -\alpha\log\Psi_t(x_t)$, using \eqref{eq:Psi_game3}, we get the desired result \eqref{eq:value_f2_game}. Solving \eqref{eq:saddle-point} (taking the gradient of \eqref{eq:value_f2_game} with respect to $x$), we get the saddle-point policies \eqref{eq: u_star_game} and \eqref{eq: theta_star_game}.
\end{proof}

 Since the right-hand side of \eqref{eq:value_f2_game} contains the expectation operation with respect to $Z$, using the strong law of large numbers \cite{durrett2019probability}, we can prove that as $N\rightarrow\infty$,
\begin{equation}\label{eq: eq:value_f2_game_MC}
    -\alpha \log \left[\frac{1}{N}\sum_{i=1}^N \exp \left\{-\frac{1}{\alpha}\int_t^T \ell_s(x^i_s)ds\right\}\right]\overset{a.s.} {\rightarrow} V_t(x_t)
\end{equation}
where $\{x_s^i, t\leq s\leq T\}_{i=1}^N$ are randomly drawn sample paths from distribution $Z$. Since $Z$ is the measure in which both $u_s$ and $\theta_s$, $t\leq s\leq T$ are zero, generating such a sample ensemble is easy.
It suffices to perform $N$ independent simulations of the dynamics $dx_s=f_s(x_s)ds + h_s(x_s)dw_s, \; x_t=x$. 

\begin{remark}\label{rem: equivalence}
Notice that the solution obtained by solving the risk-sensitive control problem \eqref{eq:value_f2_risk} is the same as the one obtained by solving the two-player zero-sum stochastic differential game \eqref{eq:value_f2_game} if $\alpha = \gamma$.
From \eqref{eq: linearizability_risk}, \eqref{eq: gamma}, and \eqref{eq: linearizability_game} we can indeed prove that $\alpha = \gamma$ i.e., the required condition to apply the path integral method to solve the risk-sensitive control problems is the same as the one in the two-player zero-sum stochastic differential games.    
\end{remark}

Similar to \eqref{eq: eq:value_f2_game_MC}, the saddle-point policies $u_t^*, \theta_t^*$ can be readily computed by the same simulated ensemble $\{x_s^i, t\leq s\leq T\}_{i=1}^N$ under distribution $Z$ and their path costs \cite{patil2023risk}. According to the strong law of large numbers, as $N\rightarrow\infty$, 

\begin{equation}\label{eq: u_star_game_MC} \mathcal{H}^u_t(x_t)\frac{\frac{1}{N}\sum_{i=1}^N\text{exp}{\left\{-\frac{1}{\alpha}\int_t^T \ell_s(x_s^i)ds\right\}}h_t(x_t)\epsilon}{\frac{\sqrt{\Delta t}}{N}\sum_{i=1}^N\text{exp}{\left\{-\frac{1}{\alpha}\int_t^T \ell_s(x_s^i)ds\right\}}} \overset{a.s.} {\rightarrow} u_t^*
\end{equation}

\begin{equation}\label{eq: theta_star_game_MC} \mathcal{H}^\theta_t(x_t)\frac{\frac{1}{N}\sum_{i=1}^N\text{exp}{\left\{-\frac{1}{\alpha}\int_t^T \ell_s(x_s^i)ds\right\}}h_t(x_t)\epsilon}{\frac{\sqrt{\Delta t}}{N}\sum_{i=1}^N\text{exp}{\left\{-\frac{1}{\alpha}\int_t^T \ell_s(x_s^i)ds\right\}}} \overset{a.s.} {\rightarrow} \theta_t^*
\end{equation}
where $\epsilon\sim\mathcal{N}(0,1)$ and $\Delta t$ is the step size. Equation \eqref{eq: u_star_game_MC} implies that under the assumption \eqref{eq: linearizability_risk} or \eqref{eq: linearizability_game}, the minimax KL control problem (Problem \ref{prob: game}) can be solved using the path integral approach. That is, if the controller has a simulator engine that can generate a large number of sample trajectories $\{x_t^i\}_{i=1}^N$ from the distribution $Z$ starting from the current state-time pair $(x_t,t)$, then the optimal attack mitigating control signal $u^*_t$ \eqref{eq: u_star_game} against the worst-case attack can be computed directly from the sample ensemble $\{x_t^i\}_{i=1}^N$. A notable feature of such a simulator-driven attack mitigation policy synthesis is that it requires neither an explicit model of the system nor an explicit policy synthesis step.

\section{Numerical Experiments}
In this section, we present two numerical studies illustrating the proposed attack synthesis and mitigation approach.
\subsection{Collision Avoidance with the Unsafe Region}
 \begin{figure*}
     \centering
       \begin{tabular}{c c}
\!\!\!\!\!\includegraphics[scale=0.37]{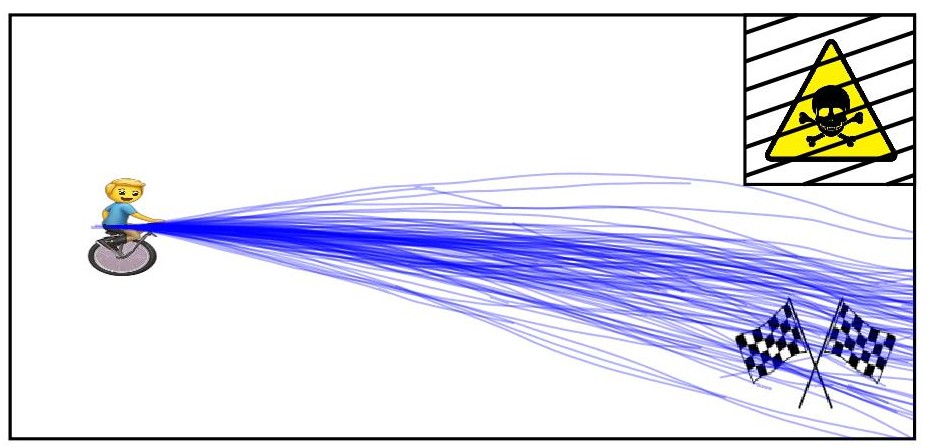} &\!\!\!\!\!\!\includegraphics[scale=0.37]{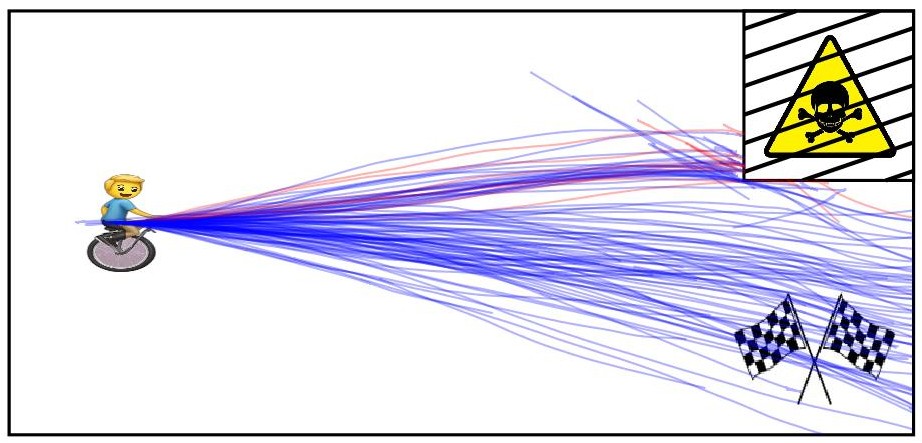} \\
\!\!\!\!\!(a) No attack, $P^{\text{crash}} \approx 0$  &\!\!\!\!\!\!(b) Stealthy attack, $\lambda=2$, $P^{\text{crash}} \approx 0.09$ \\

\!\!\!\!\!\includegraphics[scale=0.37]{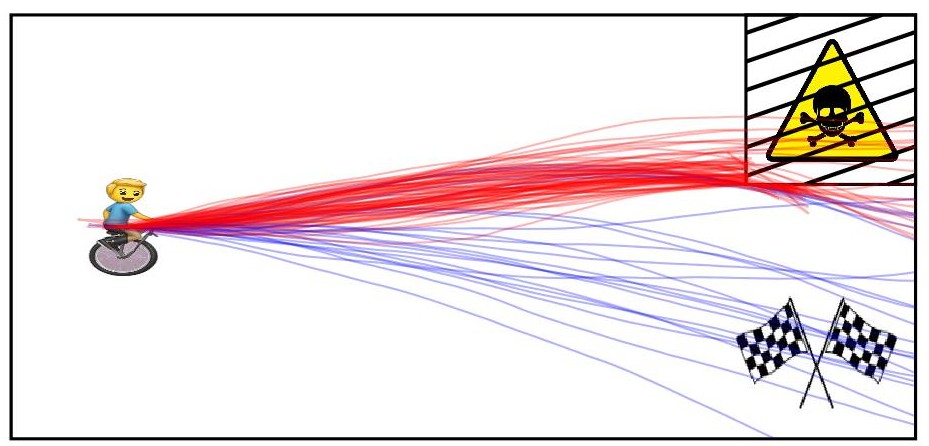} &\!\!\!\!\!\!\includegraphics[scale=0.37]{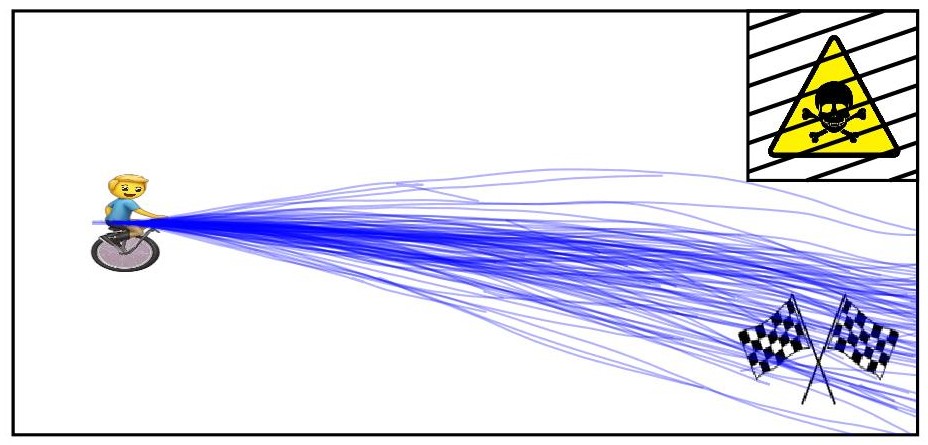} \\
\!\!\!\!\!(c) Stealthy attack, $\lambda=0.1$, $P^{\text{crash}} \approx 0.78$  &\!\!\!\!\!\!(d) Attack mitigation, $\lambda=0.1$, $P^{\text{crash}} \approx 0$ \\

       \end{tabular}
         \caption{A unicycle navigation problem. $100$ sample paths generated without the attacker and with the attacker for two values of $\lambda$ are shown. The probability of crashed paths $P^{\text{crash}}$ are noted below each case.} 
         \label{Fig. worst-case attack}
 \end{figure*}
 
Consider an autonomously operated unicycle that receives control commands $u_t$ to move from an initial location to a target destination. In this setting, a noise signal $w_t$ perturbs the control commands, and a stealthy attacker may hijack control by injecting an attack signal $v_t\neq w_t$, aiming to misguide the unicycle to an unsafe region. Consider the following unicycle dynamics model:
\begin{equation*} \label{unicycle model}
\begin{aligned}
    \begin{bmatrix}
    d{p}^x_t\\d{p}^y_t\\d{s_t}\\d{\phi_t}
    \end{bmatrix}\!\!=&\!
    \begin{bmatrix}
    {s}_t\cos{{\phi_t}}\\{s_t}\sin{{\phi_t}}\\0\\0
    \end{bmatrix}\!dt+\!\!\begin{bmatrix}
    0 & 0\\0 & 0\\1 & 0\\0 & 1
    \end{bmatrix} \!
    \begin{bmatrix}
    a_t\\
    \omega_t
    \end{bmatrix}\!dt \\
    & +\!
    \begin{bmatrix}
    0 & 0\\0 & 0\\1 & 0\\0 & 1
    \end{bmatrix} \!
    \begin{bmatrix}
    \sigma_t & 0\\
    0 & \nu_t 
    \end{bmatrix}\left(
    \begin{bmatrix}
    \Delta a_t\\
    \Delta \omega_t
    \end{bmatrix}\!dt \!+\!
    d{w}_t\right)
   ,
\end{aligned}
\end{equation*}
where $p_t:=\begin{bmatrix}
  {p}^x_t & {p}^y_t  
\end{bmatrix}^\top$, ${s}_t$ and ${\phi_t}$ denote the position, speed, and the heading angle of the unicycle, respectively, at time $t$. The control input $u_t:=\begin{bmatrix} a_t & \omega_t \end{bmatrix}^\top$ consists of acceleration $a_t$ and angular speed $\omega_t$. $\theta_t:=\begin{bmatrix} \Delta a_t & \Delta \omega_t \end{bmatrix}^\top$ is the attacker's bias input, $d{w}_t\in\mathbb{R}^2$ is the white noise and $\sigma_t$, $\nu_t$ are the noise level parameters. As illustrated in Figure \ref{Fig. worst-case attack}, the unicycle’s objective is to navigate from her initial position to the target position $\begin{bmatrix}
    \mathcal{G}^x & \mathcal{G}^y
\end{bmatrix}^\top$ (shown by two cross flags at the bottom right). Let $\mathcal{X}^\text{unsafe}$ denote the unsafe region shown by a hatched box at the top right. \par

\subsubsection{Stealthy Attack Synthesis}
First, we will compute the worst-case attack signal (i.e., solve Problem \ref{prob: KL}) assuming that the controller's policy $u_t$ is fixed and is known to the attacker. $u_t$ is designed to drive the unicycle to the target location. For the simulation, we set $\sigma_t=\nu_t=0.1$, $\forall t$, $T=5$, $c_t({x_t, u_t}) = b_t\left[\left(\mathcal{G}^x-{p}^x_t\right)^2 + \left(\mathcal{G}^y-{p}^y_t\right)^2\right] + \frac{1}{2}u_t^\top u_t + \eta_t\mathds{1} _{{p_t}\in\mathcal{X}^{\text{unsafe}}}$ where $b_t = 0.1, \eta_t = 0.1, \forall t$. $\mathds{1} _{{p_t}\in\mathcal{X}^{\text{unsafe}}}$ is an indicator function which returns $1$ when the position $p_t$ of the unicycle is inside the unsafe region $\mathcal{X}^\text{unsafe}$ and $0$ otherwise. In order to evaluate the optimal bias input \eqref{eq: theta_star} via Monte Carlo sampling, $10^4$ trajectories and a step size equal to $0.01$ are used. Figure \ref{Fig. worst-case attack}(a) shows the plot of $100$ trajectories when the system is under no attack i.e., the bias input $\theta_t=0, \forall t$. The trajectories are color-coded; the red paths crash with the unsafe region $\mathcal{X}^{\text{unsafe}}$, while the blue paths converge in the neighborhood of the target position. Figures \ref{Fig. worst-case attack}(b) and \ref{Fig. worst-case attack}(c) show the plots of $100$ trajectories with the same color-coding scheme, when the system is under the attack for two values of $\lambda$. A lower value of $\lambda$ implies that the attacker cares less about being stealthy and more about crashing the unicycle with the unsafe region $\mathcal{X}^{\text{unsafe}}$. A higher value of $\lambda$ implies the opposite. We also report $P^{\text{crash}}$, the percentage of paths that crash with the unsafe region $\mathcal{X}^{\text{unsafe}}$. Under no attack (Figure \ref{Fig. worst-case attack}(a)), none of the paths crash with the unsafe region. On the other hand, some paths crash under the adversary's attack, and for a lower value of $\lambda$, $P^{\text{crash}}$ is higher.

\subsubsection{Attack Risk Mitigation}
Now, we will solve the controller's problem who is interested in mitigating the risk of stealthy attacks i.e., solve Problem \ref{prob: minimax_KL}. In Section \ref{sec: Connections with Risk-Sensitive Control and Two-Player Zero-Sum Stochastic Differential Game}, we established the connections between the minimax KL control problem (Problem \ref{prob: minimax_KL}) with risk-sensitive control (Problem \ref{prob: risk-sensitive control}) and two-player zero-sum stochastic differential game (Problem \ref{prob: game}). We showed in Section \ref{sec: Two-player Zero-Sum Stochastic Differential Game via Path Integral Approach} that under Assumption \ref{Assumption: linearity_risk} or \ref{Assumption: linearity}, the solutions of both the problems—Problems \ref{prob: risk-sensitive control} and \ref{prob: game}—lead to the same optimal policy of the controller \eqref{eq: u_star_game} under the worst-case attack signal. In order to use the path integral control to solve Problem \ref{prob: risk-sensitive control} or \ref{prob: game}, it is necessary to find a constant $\alpha>0$ (by Assumption \ref{Assumption: linearity}) such that 
\begin{equation*}
    \begin{bmatrix}
    \sigma_t & 0\\
    0 & \nu_t 
    \end{bmatrix}\begin{bmatrix}
    \sigma_t & 0\\
    0 & \nu_t 
    \end{bmatrix}^\top = \alpha\left(I_{2\times2} - \frac{1}{\lambda}\begin{bmatrix}
    \sigma_t & 0\\
    0 & \nu_t 
    \end{bmatrix}\begin{bmatrix}
    \sigma_t & 0\\
    0 & \nu_t 
    \end{bmatrix}^\top\right)
\end{equation*}
where $I_{2\times2}$ is an identity matrix of size $2\times2$. We solve Problem \ref{prob: game} and evaluate the saddle-point policies \eqref{eq: u_star_game}, \eqref{eq: theta_star_game} via path integral control. Note that solving Problem \ref{prob: risk-sensitive control} will also give us the same results. We use $10^4$ Monte Carlo trajectories and a step size equal to $0.01$. Figure \ref{Fig. worst-case attack}(d) shows the plots of $100$ sample trajectories generated using synthesized saddle-point policies $(u^*_t, \theta^*_t)$ for $\lambda = 0.1$. In Figures \ref{Fig. worst-case attack}(b) and \ref{Fig. worst-case attack}(c), the controller is unaware of the attacker. However, in Figure \ref{Fig. worst-case attack}(d), the controller is aware of the attacker and designs an attack mitigating policy to combat the potential attacks. As we observe, under the attack mitigating policy, the controller is able to avoid the crashes with the unsafe region $\mathcal{X}^{\text{unsafe}}$.  

\subsection{Cruise Control}
\begin{figure*}
     \centering
       \begin{tabular}{c c}
\!\!\!\!\!\includegraphics[scale=0.18]{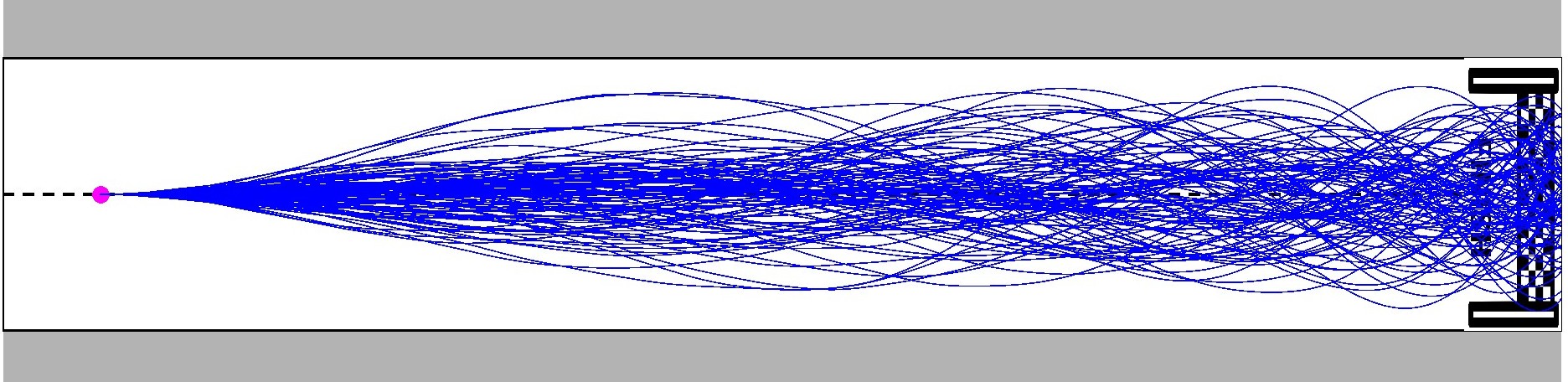} &\includegraphics[scale=0.18]{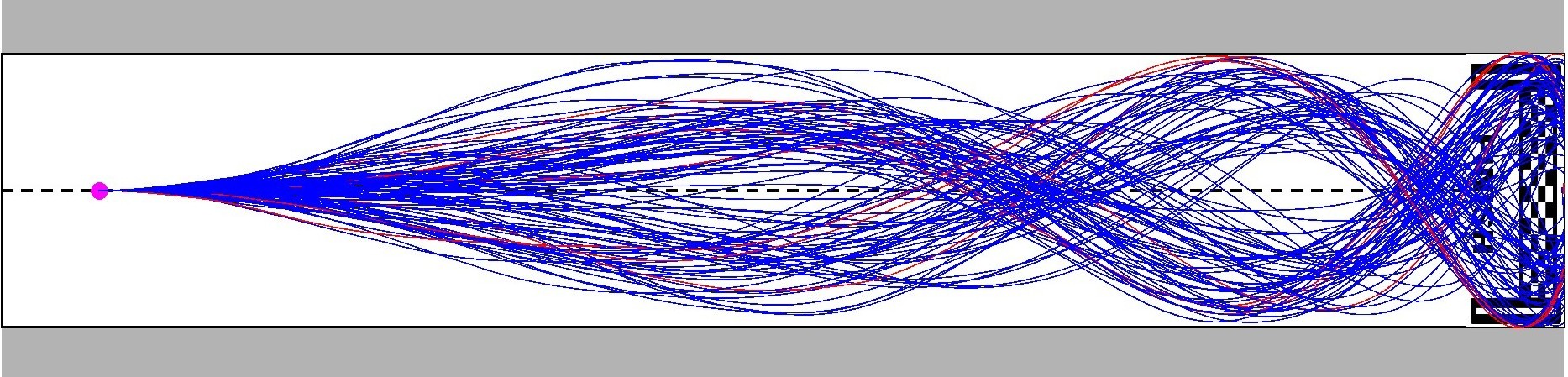} \\
\!\!\!\!\!(a) No attack, $P^{\text{crash}} \approx 0$  &(b) Stealthy attack, $\lambda=3$, $P^{\text{crash}} \approx 0.07$ \\\\

\!\!\!\!\!\includegraphics[scale=0.18]{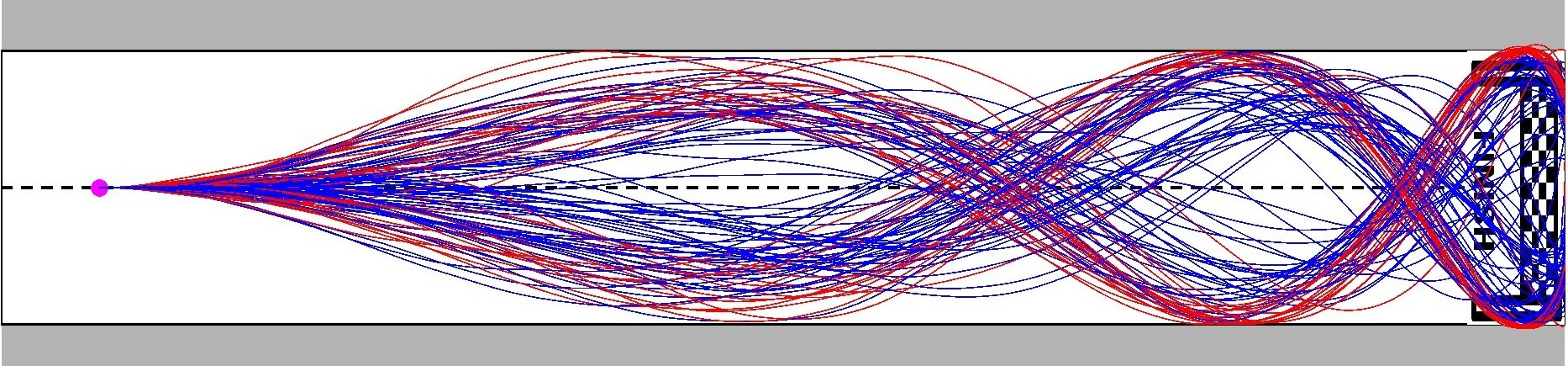} &\includegraphics[scale=0.18]{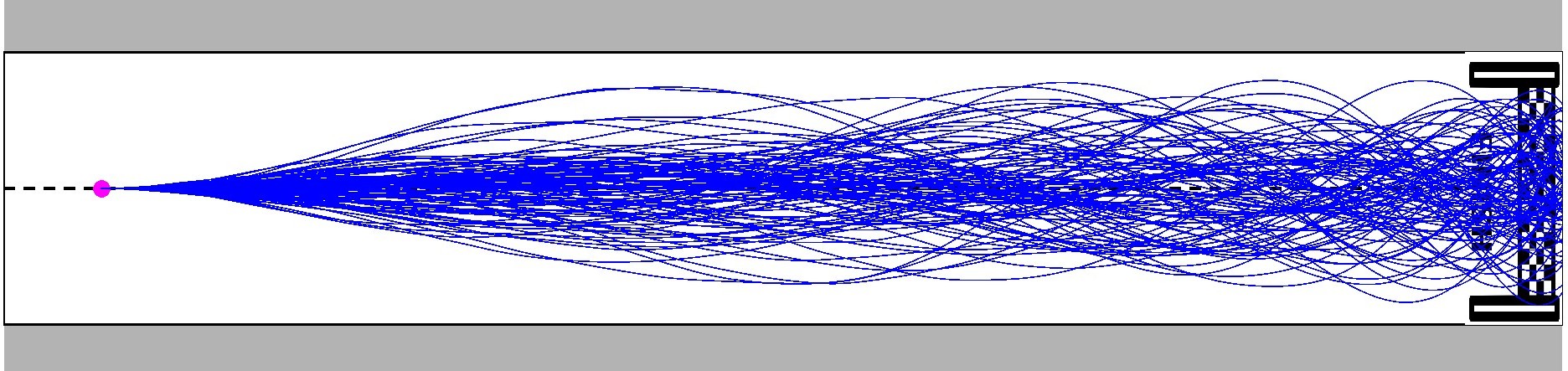} \\
\!\!\!\!\!(c) Stealthy attack, $\lambda=1.5$, $P^{\text{crash}} \approx 0.35$  &(d) Attack mitigation, $\lambda=1.5$, $P^{\text{crash}} \approx 0$ \\

       \end{tabular}
         \caption{Cruise control problem. A magenta dot on the left shows the start position, and the finish line is shown on the right (centered Y-intersection of the flag). $100$ sample paths generated without the attacker and with the attacker for two values of $\lambda$ are shown. The probability of crashed paths $P^{\text{crash}}$ are noted below each case.} 
         \label{Fig. car_worst-case attack}
 \end{figure*}
 
Consider an autonomously controlled car that receives control commands $u_t$ to remain on the road without veering off track. In this setup, a noise signal $w_t$ perturbs the control commands, and a stealthy attacker may hijack control by injecting an attack signal $v_t$, aiming to drive the car off the track. Consider the following car model:

\begin{equation*} \label{car model}
\begin{aligned}
    \begin{bmatrix}
    d{p}^x_t\\d{p}^y_t\\d{s_t}\\d{\delta_t}\\d\phi_t
    \end{bmatrix}\!=&
    \begin{bmatrix}
    {s_t}\cos{{\delta}_t}\\{s_t}\sin{{\delta}_t}\\0\\\frac{s_t\tan{\phi}_t}{L}
    \\0
    \end{bmatrix}\!dt +\begin{bmatrix}
    0 & 0\\0 & 0\\1 & 0\\0&0\\0 & 1
    \end{bmatrix} \!
    \begin{bmatrix}
    a_t\\
    \zeta_t
    \end{bmatrix}\!dt \\
    &\!\!\!\!+\begin{bmatrix}
    0 & 0\\0 & 0\\1 & 0\\0 & 0\\0 & 1
    \end{bmatrix}\begin{bmatrix}
    \sigma_t & 0\\
    0 & \nu_t 
    \end{bmatrix} \!
    \left(\begin{bmatrix}
    \Delta a_t\\
    \Delta\zeta_t
    \end{bmatrix}dt + d{w}_t
    \right)
\end{aligned}
\end{equation*}
where $p_t:=\begin{bmatrix}
  {p}^x_t & {p}^y_t  
\end{bmatrix}^\top$, ${s}_t$, ${\delta_t}$ and $\phi_t$ denote the position, speed, and the heading angle and the front wheel angle of the car, respectively, at time $t$. $L$ is the inter-axle distance. The control input $u_t:=\begin{bmatrix} a_t & \zeta_t \end{bmatrix}^\top$ consists of acceleration $a_t$ and the front wheel angular rate $\zeta_t$. $\theta_t:=\begin{bmatrix} \Delta a_t & \Delta \zeta_t \end{bmatrix}^\top$ is the attacker's bias signal, $d{w}_t\in\mathbb{R}^2$ is the white noise and $\sigma_t$, $\nu_t$ are the noise level parameters. As illustrated in Figure \ref{Fig. car_worst-case attack}, the car’s objective is to drive on the road without going off track until it crosses the finish line. Let $\mathcal{G}^{\text{center}}$ denote the center of the road, $\mathcal{G}^{\text{finish}}$ denote the finish line, and $\mathcal{X}^\text{unsafe}$ denote the off-the-track region shown in gray.

\subsubsection{Stealthy Attack Synthesis}
First, we will compute the worst-case attack signal assuming that the controller's policy $u_t$ is fixed and is known to the attacker. $u_t$ is designed to drive the car to the finish line without going off track. For the simulation, we set $L=0.05, \sigma_t=\nu_t=\sqrt{0.005}$, $\forall t$, $T=10$, $c_t({x_t, u_t}) = b_t\left[\left(\mathcal{G}^{\text{finish}}-{p}^x_t\right)^2 + \left(\mathcal{G}^{\text{center}}-{p}^y_t\right)^2\right] + \frac{1}{2}u_t^\top R_t u_t + \eta_t\mathds{1} _{{p_t}\in\mathcal{X}^{\text{unsafe}}}$ where $b_t = 0.02, R_t =4, \eta_t = 0.02, \forall t$. $\mathds{1} _{{p_t}\in\mathcal{X}^{\text{unsafe}}}$ is an indicator function which returns $1$ when the car goes off track and $0$ otherwise. In order to evaluate the optimal bias input \eqref{eq: theta_star} via Monte Carlo sampling, $10^4$ trajectories and a step size equal to $0.02$ are used. Figure \ref{Fig. car_worst-case attack}(a) shows the plot of $100$ trajectories when the car is under no attack i.e., the attack signal $\theta_t=0, \forall t$. The trajectories are color-coded; the red paths go off track, while the blue paths cross the finish line without going off track. Figures \ref{Fig. car_worst-case attack}(b) and \ref{Fig. car_worst-case attack}(c) show the plots of $100$ trajectories with the same color-coding scheme, when the system is under the attack for two values of $\lambda$. A lower value of $\lambda$ implies that the attacker cares less about being stealthy and more about driving the car off track. We also report $P^{\text{crash}}$, the percentage of paths that go off track. Under no attack (Figure \ref{Fig. car_worst-case attack}(a)), none of the paths go off track. On the other hand, trajectories under the adversary's attack drift off from the center of the road, and some of them go off track. Also, notice that for a lower value of $\lambda$, $P^{\text{crash}}$ is higher. 

\subsubsection{Attack Risk Mitigation}
Now, we will solve the controller's problem who is interested in mitigating the risk of stealthy attacks i.e., solve Problem \ref{prob: minimax_KL}. In order to use the path integral control to solve Problem \ref{prob: risk-sensitive control} or \ref{prob: game}, it is necessary to find a constant $\alpha>0$ (by Assumption \ref{Assumption: linearity}) such that 
\begin{equation*}
    \begin{bmatrix}
    \sigma_t & 0\\
    0 & \nu_t 
    \end{bmatrix}\begin{bmatrix}
    \sigma_t & 0\\
    0 & \nu_t 
    \end{bmatrix}^\top = \alpha\left(R_t^{-1} - \frac{1}{\lambda}\begin{bmatrix}
    \sigma_t & 0\\
    0 & \nu_t 
    \end{bmatrix}\begin{bmatrix}
    \sigma_t & 0\\
    0 & \nu_t 
    \end{bmatrix}^\top\right).
\end{equation*}
We solve Problem \ref{prob: game} and evaluate the saddle-point policies \eqref{eq: u_star_game}, \eqref{eq: theta_star_game} via path integral control. We use $10^4$ Monte Carlo trajectories and a step size equal to $0.02$. Figure \ref{Fig. car_worst-case attack}(d) shows the plots of $100$ sample trajectories generated using synthesized saddle-point policies $(u^*_t, \theta^*_t)$ for $\lambda = 1.5$. In Figures \ref{Fig. car_worst-case attack}(b) and \ref{Fig. car_worst-case attack}(c), the controller is unaware of the attacker. However, in Figure \ref{Fig. car_worst-case attack}(d), the controller is aware of the attacker and designs an attack-mitigating policy to combat the potential attacks. As we observe, under the attack mitigating policy, the controller is able to cross the finish line without going off-track.

\section{Conclusion and Future Work}\label{sec:conclusion}
We presented a framework for understanding and counteracting stealthy attacks in continuous-time, nonlinear cyber-physical systems. By representing the trade-off between remaining covert and degrading system performance through a Kullback–Leibler (KL) divergence measure of stealthiness, we formulated both (i) a KL control problem to characterize the attacker’s worst-case policy, and (ii) a minimax KL control problem to design a controller that mitigates the risk of such attacks. Key to our approach is a simulator-driven path integral control method that uses forward Monte Carlo simulations rather than closed-form models, thus remaining viable even for high-dimensional or complex systems. Numerical experiments on a unicycle navigation and cruise control illustrated that the attacker can indeed degrade the system performance while maintaining low detectability, yet the controller can act against the attack signals through risk-aware policy design.\par
Going forward, our path integral–based methodology can be extended to broader settings, such as partial observability or more general cost functions, and can accommodate advanced learning-based models that serve as ``digital twins" of physical processes. This work thus provides both a theoretical and computational foundation for tackling stealthy attack synthesis and mitigation in modern CPS. We also plan to conduct the sample complexity analysis of the path integral approach to solve both the KL control and the minimax KL control problems.

{\appendix}
Consider the following hypothesis testing problem on a sample path $v_t, 0\leq t \leq 1$:
\begin{align*}
H_0 :& \quad dv_t=dw_t, \; v_0=0 \\
H_1 :& \quad dv_t=\sigma dw_t, \; v_0=0, \; \sigma=1.1.
\end{align*}
That is, under $H_0$, $v_t$ follows the law of the standard Brownian motion $w_t$, whereas under $H_1$, $v_t$ is slightly ``noisier." Fig.~\ref{fig:h0h1} shows ten independent sample paths of $v_t$ generated under each hypothesis.

\begin{figure}[h]
\includegraphics[scale=0.3]{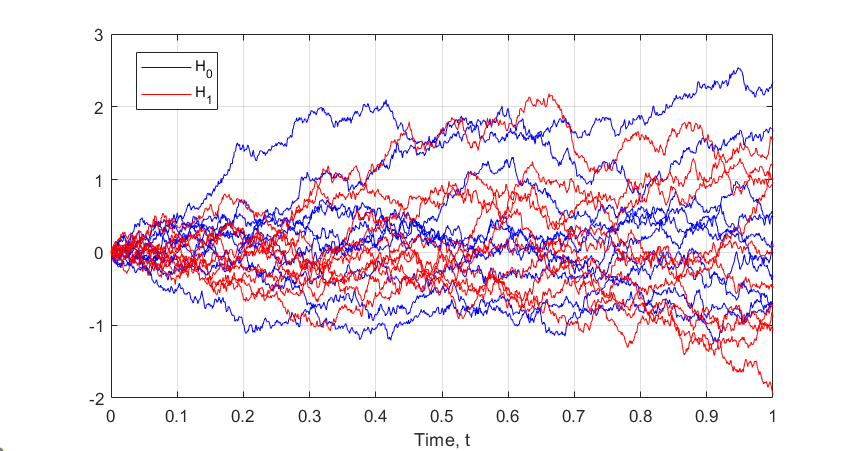}
\caption{Sample paths under $H_0$ and $H_1$.}
\label{fig:h0h1}
\end{figure}

Suppose the detector's task is to determine $H_0$ vs. $H_1$ based on discrete-time measurements $v_{t_k}, t_k=hk$, $k=1,2,\cdots, K(=1/h)$ where $h$ is the discretization step size.
If the underlying model is $dv_t=\sigma dw_t$, then the detector observes $K=(1/h)$ i.i.d. samples of increments:
\[
y_k:=v_{t_k}-v_{t_{k-1}}
\!\!\!\stackrel{\text{i.i.d.}}{\sim}\!\! \mathcal{N}(0, \sigma^2 h), \;\; k=1, 2, \cdots, K(=1/h).
\]
If we denote $P_0=\mathcal{N}(0,h)$ and $P_1=\mathcal{N}(0,\sigma^2 h)$, the likelihood ratio is
\begin{align*}
    \frac{dP_1}{dP_0}(y_1, \cdots, y_K)=&\prod_{k=1}^K \frac{\frac{1}{\sqrt{2\pi \sigma^2 h}}\exp\left(-\frac{y_k^2}{2\sigma^2 h}\right)}{\frac{1}{\sqrt{2\pi h}}\exp\left(-\frac{y_k^2}{2 h}\right)}\\
=&
\frac{1}{\sigma^K}\exp\left[\frac{1}{2h}\left(1-\frac{1}{\sigma^2}\right)\sum_{k=1}^K y_k^2\right].
\end{align*}
Thus, the acceptance region for $H_1$ under the Neyman-Pearson test is
\begin{equation}
\frac{1}{\sigma^K}\exp\left[\frac{1}{2h}\left(1-\frac{1}{\sigma^2}\right)\sum_{k=1}^K y_k^2\right] \geq \tau
 \label{eq:h1_region}
\end{equation}
where $\tau$ is a threshold.

Let us represent the probability of type-I error (false alarm) by $\alpha$ and the probability of type-II error (failure of detection) by $\beta$. In this example, both $\alpha$ and $\beta$ can be computed analytically. 
The false alarm rate $\alpha$ is the probability of the event \eqref{eq:h1_region} when $y_k\stackrel{\text{i.i.d.}}{\sim} \mathcal{N}(0, h)$.
Notice that \eqref{eq:h1_region} is equivalent to
\[
\sum_{k=1}^K \left(\frac{y_k}{\sqrt{h}}\right)^2 \geq \frac{2\sigma^2}{\sigma^2-1}(K\log \sigma + \log \tau)
\]
and $\sum_{k=1}^K \left(\frac{y_k}{\sqrt{h}}\right)^2\sim \chi^2(K)$ (the $\chi^2$-distribution with $K$ degrees of freedom).
Since the CDF of $\chi^2(K)$ is
\[
F(x)=P\left(\frac{x}{2},\frac{K}{2}\right)=\frac{1}{\Gamma\left(\frac{K}{2}\right)}\int_0^{x/2} t^{K/2-1}e^{-t}dt
\]
where $\Gamma(a)$ is the Gamma function and $P(x,a)$ is the regularized lower incomplete Gamma function, we have
\begin{align}
\alpha =& 1-P\left(\frac{\sigma^2}{\sigma^2-1}(K\log \sigma + \log \tau), \frac{K}{2}\right)\nonumber\\ =&Q\left(\frac{\sigma^2}{\sigma^2-1}(K\log \sigma + \log \tau), \frac{K}{2}\right)
 \label{eq:alpha}
\end{align}
where $Q(x,a)$ is the regularized upper incomplete Gamma function.

On the other hand, the detection failure rate $\beta$ is the probability of the event
\[
\sum_{k=1}^K \left(\frac{y_k}{\sqrt{\sigma^2 h}}\right)^2 < \frac{2}{\sigma^2-1}(K\log \sigma + \log \tau)
\]
when $y_k\stackrel{\text{i.i.d.}}{\sim} N(0, \sigma^2 h)$.
Since $\sum_{k=1}^K \left(\frac{y_k}{\sqrt{\sigma^2 h}}\right)^2\sim \chi^2(K)$, we have
\begin{equation}
\beta = P\left(\frac{1}{\sigma^2-1}(K\log \sigma + \log \tau), \frac{K}{2}\right). \label{eq:beta}
\end{equation}
Using \eqref{eq:alpha} and \eqref{eq:beta}, the trade-off curve between $\alpha$ and $\beta$ can be plotted by varying the parameter $\tau$.

Fig.~\ref{fig:alpha_beta} shows the $\alpha$--$\beta$ trade-off curve when the discretization step size is set to be $h=1/100, 1/200$ and $1/300$. It shows that the trade-off curves are drastically different if different step sizes are chosen. It can be numerically verified that both $\alpha$ and $\beta$ can be made arbitrarily small simultaneously by taking $h\searrow 0$.

\begin{figure}[h]
\centering
\includegraphics[width = 0.45\textwidth]{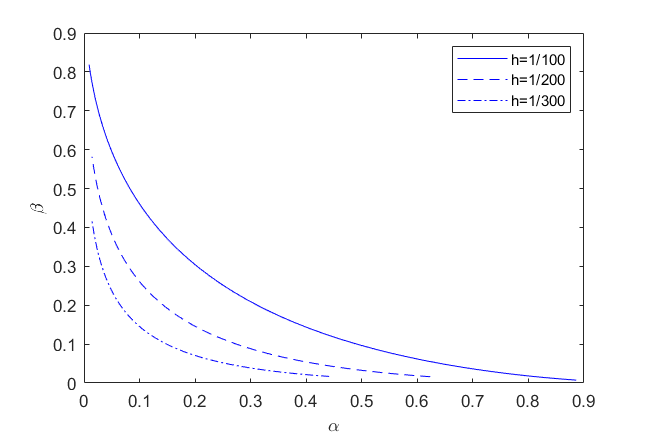}
\caption{$\alpha$--$\beta$ trade-off curves.}
\label{fig:alpha_beta}
\end{figure}

This observation implies that the hypothesis testing problem formulated above is ill-posed -- that is, the detector who can observe a continuous-time sample path $v_t, 0\leq t \leq 1$ (and hence can choose an arbitrarily small $h$) can correctly estimate which model has generated the path with probability one.
Mathematically, this reflects the fact that the probability models $H_0$ and $H_1$ are mutually singular. 
Specifically, let $\mu_0$ and $\mu_1$ be the probability measures under which $v_t$ in $H_0$ and $H_1$ are the standard Brownian motions, respectively. It can be shown that $\mu_0$ and $\mu_1$ are mutually singular, i.e., the Radon-Nikodym derivative $\frac{d\mu_1}{d\mu_0}$ is ill-defined.

\bibliographystyle{IEEEtran}
\bibliography{references}
\end{document}